%% file: main.tex
\documentclass[a4paper,USenglish,cleveref,autoref,thm-restate]{amsart}

\usepackage[letterpaper,top=2cm,bottom=2cm,left=3cm,right=3cm,marginparwidth=1.75cm]{geometry}

\usepackage{amsthm}
\usepackage{amsmath}
\usepackage{amsfonts}
\usepackage{mathrsfs}
\usepackage{graphicx}
\usepackage{enumerate}
\usepackage{multicol}
\usepackage[colorlinks=true, allcolors=blue]{hyperref}

\usepackage{algorithm}
\usepackage{algpseudocode}

\numberwithin{equation}{section}
\newtheorem{proposition}[equation]{Proposition}
\newtheorem{theorem}[equation]{Theorem}
\newtheorem{lemma}[equation]{Lemma}
\theoremstyle{definition}
\newtheorem{definition}[equation]{Definition}
\newtheorem{example}[equation]{Example}

\usepackage{ourmacros}
\usepackage{paperinfo}

\begin{document}

\begin{abstract}
    \input{sections/abstract}
\end{abstract}

\maketitle

\section{Introduction}

Complex networks \cite{Easley2010, Barabasi2016} arise in many applications, such as road networks, circuits, social networks, and neuroscience. 
Topological data analysis (TDA) \cite{Dey2021, Ghrist2014, Munch2017} and particularly persistent homology \cite{Edelsbrunner2013,Kerber2016}, provides tools for measuring the overall shape of an input structure, including many methods for the case of where this network data can be encoded in a graph structure; see \cite{Aktas2019} for a survey in this special case.
The choice of TDA representation of the graph data depends on the particular properties of the graph; in particular, graphs may be directed or not, and may or may not have weights on the edges. 
However, homology, and thus persistent homology, takes as input undirected simplicial complexes, so methods for incorporating directionality pose a particular challenge. 
The current methods for studying these objects are often either mathematically precise but computationally impractical or miss aspects of the directed structure.
Still, it is tempting to find ways to replace the complex input structure, such as a directed graph, by a simpler representation, such as a persistence diagram, where additional structure such as metrics are available for comparison. 
In particular, this opens the door to tasks such as classification of input graphs, where we can try to automatically attach a label to an input graph given the labels of graphs with similar persistence diagrams, e.g.~\cite{Myers2019,Myers2023,Myers2023a,Immonen2023,Horn2022,Hofer2020,Rieck2019,Carriere2020,Chen2023}. 

In this paper, we use the terminology that a graph $G = (V,E)$ is a finite collection of vertices $V$ with edge set $E \subseteq V \times V$. 
This input can either be undirected, where $(u,v) = (v,u)$; else it is a directed graph, also known as a digraph. 
In either case, it is weighted if there is also a function $\omega: E \to \R$ and then we denote the structure as $G = (V,E,\omega)$. 
Following \cite{Carlsson2013,Chowdhury2023}, in this paper we use the term \textit{network} to mean a complete directed weighted graph, $(V,V \times V, \omega)$. 
This information can also be viewed as arising from an asymmetric function $\omega:V\times V \rightarrow \R$, and we will assume that we are given non-negative weights taking values in $\R_{\geq 0}$ and that $\omega(x,x')$ is zero if and only if $x=x'$.

We next survey methods for encoding network data, or more generally weighted digraph data, into something which can be measured with persistent homology. 
Arguably, the most prominent tool for this is Dowker persistence \cite{CM2018-Dowker}, arising from the Dowker complex of a relation \cite{Dowker1952}. 
In the setting of an input weighted digraph $G$, a Dowker sink complex is built by including a simplex in an undirected complex for every collection that is witnessed by a vertex; i.e., $(v_0,\cdots,v_n)$ is a simplex in a complex $K$ iff there is a vertex $v'$ with directed edges $(v_i,v')$ for all $i$ in the graph $G$. 
A similar version, called the Dowker source complex, can be defined replacing the sink $v_i$ with a source $v_i$ which has edges from (rather than to) all the vertices of the simplex. 
This complex is filtered by restricting to edges with weight at most $\delta$ for some $\delta \geq 0$,  and then standard persistent homology can be applied. 
In \cite{CM2018-Dowker}, Chowdhury and M\'emoli proved that while the sink and source complexes are not the same in general, the resulting persistence diagrams are indeed the same.
This complex is computable in polynomial time, making it practical for use in applications; e.g.~\cite{Chowdhury2016a} and \cite{Liu2022}.
Further, this construction is stable in the sense that for the $\infty$-network distance defined in \cite{Carlsson2010},  \cite{CM2018-Dowker} showed that the distance between input networks gives an upper bound for the bottleneck distance between the resulting persistence diagram. 
As we will give detailed comparisons with our results with Dowker persistence, further details on this construction will be given in Section \ref{ssec:dowker}. 
Additionally, there is some recent work implementing Dowker filtrations into graph neural networks to approximate persistent homology of directed graphs \cite{NeuralDowkerNetwork2024}.

Another collection of available methods for turning digraph data into something which can be measured with persistent homology uses ordered tuple complexes. 
In this case, one can redefine homology rather than constructing a simplicial complex. %
In this case, the $n$-chains of an \textit{ordered chain complex} (see e.g.~\cite[Page~76]{Munkres2}) are defined by ordered $n$-tuples given as $(v_0,\cdots,v_n)$, where the $v_i$ are vertices
but might not be distinct.
In~\cite{Turner2019a}, four different persistence modules are defined in terms of a real-valued function $f:X\times X\rightarrow \R$.
The first two include a symmetrization of the function and a construction by ordered tuple complexes; these two become equivalent in the case where $f$ is a symmetric function.
The other two methods use a filtration of directed graphs, either to encode strongly connected components or to build a filtration of ordered tuple complexes.
In \cite{Mendez2023}, they construct a new homology theory based on a directed simplicial complex (another name for ordered tuple complex) to encode exclusively directed cycles.
All of these methods give robust theoretical approaches to describe the asymmetry information.
However, because repeated vertices are allowed, there is an issue with computational blow-up for these complexes resulting in their limited use in practice.

A similar idea arises from using the theory of path homology \cite{Grigoryan2013} in directed graphs,
which builds a chain complex on the collection of unweighted length-$k$ paths, restricting to those paths for which the boundary map returns a valid path.
The persistent path homology (PPH) \cite{CM2018-PPH} filters the directed graph (or network) based on edge weights to obtain a filtration of directed graphs, where path homology is then utilized. 
Advances on 1-dimensional PPH computation were made in \cite{Dey2020}.
Furthermore, a modification of this method, called Grounded PPH \cite{Chaplin2024}, gives a method with more geometrically interpretable results, reflecting on its robustness under modifications of the input graph such as edge subdivisions, deletions, and collapses.

Finally, one can construct a directed flag complex, generalizing the more common construction seen in the undirected case \cite{Aharoni2005}.
A collection $(v_0,\cdots,v_n)$ is a directed $n$-simplex if for every $0 \leq i<j\leq n$ there is a directed edge from $v_i$ to $v_j$.
The idea of a directed flag complex was introduced for use in neuroscience data in \cite{Reimann2017} and is further studied, for example, in \cite{Conceicao2022}, as well as in terms of homotopy in \cite{Govc2020} and persistent homology in \cite{Govc2021, Ignacio2019}.
This complex can be computed using the \texttt{Flagser} package \cite{Luetgehetmann2020}.
Furthermore, a computational study of `almost $d$-simplices', that is, a directed simplex which is missing an edge, can be found in \cite{Unger2023}.

One particularly useful property of any of these filtration methods is that of stability, i.e.,~if the input data is close in some metric, then the resulting representation is closer in its own metric.
The main distance measure in this field to compare the input directed graph data is that of the network distance $d_\cN$ \cite[Definition 2.2.7]{Chowdhury2023}, generalizing the idea of the Gromov-Hausdorff distance \cite{Gromov2007}; in this paper, to clearly differentiate from other distances we will define, we call this the $\infty$-network distance.
This metric is used to study clustering methods in \cite{Carlsson2014, Chowdhury2016, Kim2024, Kim2020}.
Because the network distance arose from generalizing a distance for metric spaces, the main requirement for using this distance is that the input graph must be metric-like in the sense that it must be complete; a complete directed weighted graph is called a network in this literature.
A general study on network distances and the stability of some of the methods mentioned above can be found in \cite{Chowdhury2023}.

The novel contribution of this paper is to construct a new filtration of undirected simplicial complexes that arises from directed graph information, called the \emph{walk-length filtration}.
This filtration encodes information about when a walk can visit all the vertices of a particular list to determine when a simplex should be included.
Then, we can represent the structure using the persistence diagram of the input filtration. 
Because we construct a filtration on an undirected simplicial complex (rather than using non-traditional homology theories as in the case of PPH or the ordered chain complex), we will largely focus on comparing this construction to the Dowker persistence diagram.
Indeed, the proof techniques and structure throughout the paper largely follow that of \cite{CM2018-Dowker}, and thus we also showcase similar experiments that highlight the differences in the two methods. 
We show that the walk length filtration is not stable under the standard $\infty$-network distance like the Dowker filtration. 
However, we give a modified definition of the network distance,  called the 1-network distance $d_\cN^1$,  which  replaces the maximum distortion with a summation.
Under this new distance, we give a proof of stability which is network size dependent; that is that the bottleneck distance between the persistence diagrams is upper bounded by the size of the input networks times the new distance. 
Since this is not a particularly tight bound, we give a second modified version of the 1-network distance which requires a bijection between the vertex sets of the compared networks, and show that in this case, the walk length persistence is a stable in the traditional sense. 
Note that we define the walk-length filtration using a walk rather than a path that may only visit a vertex at most once; otherwise, computation time would be related to finding Hamiltonian paths, which is an NP-hard problem. 
We give a dynamic programming algorithm for its computation. 
Finally, we give examples highlighting the differences between the Dowker filtration and the walk-length filtration. 
Specifically, we replicate two studies from \cite{CM2018-Dowker}.
First, we examine the walk-length persistence of cycle networks and show that the filtration coincides with the Dowker filtration, but after making a modification in such graphs, we see that, in some cases, walk-length can identify the modification while Dowker cannot, up to 1-dimensional homology.
This demonstrates a higher sensitivity of walk-length persistence to non-directed cycles.
Second, we use the walk-length persistence for classification of simulated hippocampal networks from \cite{CM2018-Dowker}, and show that different choices of preprocessing of the weights as input to the two filtrations results in different representations and qualities of classification.

\section{Background}
\label{sec:background}

In this section, we will give the relevant background needed for defining the walk-length filtration.

\subsection{Persistent Homology}\label{sec:persistent-homology}

We describe the general construction of persistent homology; we direct the interested reader to \cite{Dey2021} for more details.

A \textit{simplicial complex} $X$ can be defined as a set of vertices $X_0$, the $0$-simplices, together with sets $X_n$ containing the \textit{$n$-simplices}, which are subsets of $X_0$ of size $n+1$.
We denote an $n$-simplex as an ordered set $\sigma=[v_0,\dots,v_n]$, as well as the even permutations of this order, while any odd permutation of the order of the vertices in $\sigma$ is then denoted as $-\sigma$.
Simplicial complexes have the property that any subset of a simplex is also a simplex, 
that is, if $\sigma\in X_n$ and $\sigma'\subset\sigma$, then $\sigma'\in X_k$ where $k+1$ is the size of $\sigma'$.
Now, given an abelian group $G$, let $C_n(X;G)$ denote the free group generated by the set of $n$-simplices $X_n$.
We write this as $C_n$ for simplicity. This group is called the group of $n$-chains; its elements are formal finite sums of the form 
$\sum_i a_i\,\sigma_i$
where $a_i\in G$ and $\sigma_i\in X_n$. We can define \textit{boundary maps} $\partial_n:C_n\rightarrow C_{n-1}$ given in the generators of $C_n$ as
\[\partial([v_0,\dots,v_n]) = \sum_{i=0}^n (-1)^i \, [v_0, \dots, \widehat{v_i}, \dots, v_n],\]
where $[v_0, \dots, \widehat{v_i}, \dots, v_n]$ denotes the $(n-1)$-simplex obtained from deleting the element $v_i$.
The boundary maps have the property that the composition $\partial_n\circ\partial_{n+1}$ is the zero map for all $n\geq2$, thus we obtain a sequence of maps, called a \textit{chain complex}, in the form
\[\dots \rightarrow C_{n+1} \xrightarrow{\partial_{n+1}} C_n \xrightarrow{\partial_n} C_{n-1} \rightarrow \dots \rightarrow C_1 \xrightarrow{\partial_1} C_0 \xrightarrow{\partial_0} 0.\]
The equation $\partial_n \partial_{n+1}=0$ is equivalent to $\Img(\partial_{n+1})\subseteq \Ker(\partial_n)$, hence we can define the \textit{$n$-th simplicial homology group of $X$ with coefficients in $G$} as the quotient group 
\[H_n(X;G) := \Ker(\partial_n) \,/\, \Img(\partial_{n+1}). \]
This is denoted just as $H_n(X)$ if the coefficients are clear from the context. Elements of $H_n(X;G)$ are cosets of $\Img(\partial_{n+1})\subset C_n$ and are called \textit{homology classes}. 

Moving on to define persistent homology, we use the coefficients $G=\K$, where $\K$ is a field, and so we think of the homology groups $H_n(X) = H_n(X;\K)$ as vector spaces over $\K$. 
Suppose that there is a filtered simplicial complex $K$ given as 
\[\emptyset = K_0 \subset K_1 \subset\dots\subset K_{m-1} \subset K_m = K.\]
Any simplicial map $K_i\rightarrow K_j$ induces a homomorphism in homology groups $H_n(K_i)\rightarrow H_n(K_j)$. 
A homology class $\alpha$ is said to be \textit{born at $K_i$} %
if it is in $H_n(K_i)$ but is not in the image of the map induced by the inclusion $K_{i-1}\subset K_i$. 
If a class $\alpha$ is born at $K_i$, it is said to \textit{die entering $K_j$} if the image of the map induced by $K_{i-1}\subset K_{j-1}$
does not contain the image of $\alpha$ but the image of the map induced by $K_{i-1}\subset K_j$ does, which is to say that $\alpha$ merged into an older class. 
The persistence of $\alpha$ is then encoded as the tuple $(i,j)$, where we set $j=\infty$ when the class never dies.

Furthermore, for each dimension $n\geq0$, the family of vector spaces $\{H_n(K_i)\}_{i=0}^m$ together with the linear maps $\{H_n(K_i)\rightarrow H_n(K_j)\}_{i<j}$ defines a \textit{persistence module}, specifically the \textit{$n$-dimensional persistence vector space of $K$}, which is denoted as $\mathcal{H}_n(K)$. This persistence module admits the decomposition
\[
\mathcal{H}_n(K) \cong \bigoplus_{\ell\in L} \mathbb{I}[p_\ell^*,q_\ell^*]
\]
as a finite direct sum of interval modules---%
an \textit{interval module} $\mathbb{I}[p^*,q^*]$ can be represented by the family of vector spaces $\{V_t\}_{t\in\R}$, with $V_t = \K$ if $t\in[p^*,q^*]$ and $V_t=0$ otherwise, together with identity maps between all pairs of nontrivial spaces, where $[p^*,q^*]$ denotes any of the intervals $[p,q],(p,q],[p,q),(p,q)$.
This decomposition is represented as a multiset, that is, a set with multiplicities, in the extended plane, called \textit{persistence diagram}, given by $\Dgm_n(K) := \{(p_\ell,q_\ell) \;:\; \ell\in L\} \subset \R \times \R\cup\{\infty\}$.
We refer to a point $(p_\ell,q_\ell)$ for which $q_\ell=\infty$ as a \textit{point at infinity}.
For a detailed construction of persistence modules, their decomposition and their diagrams, we refer the reader to \cite{Chazal2016}.

We now define the bottleneck distance between persistence diagrams. For simplicity in the definition, we adjoin countably many copies of the diagonal $\Delta:=\{(x,x): x\in\R\}$ to all diagrams.
Let $\mathscr{D} = \Dgm_n(K)$ and $\mathscr{D}' = \Dgm_n(K')$ be two persistence diagrams associated with two filtered simplicial complexes $K$ and $K'$, respectively. The $L_\infty$-distance between two points $u = (u_1,u_2)$ and $v = (v_1,v_2)$ in the extended plane is given by $\|u-v\|_\infty = \max\{|u_1-v_1|, |u_2-v_2|\}$, where
the difference between two infinite coordinates is defined as zero. 
The \textit{bottleneck distance} is then defined as 
\[d_B(\Dgm_p(K),\Dgm_p(K')) = \inf_\eta \sup_{x\in \mathscr{D}} \|x-\eta(x)\|_\infty,\]
where the infimum is taken over all bijections $\eta:\mathscr{D}\rightarrow \mathscr{D}'$. Note that adjoining the diagonal is required because two persistence diagrams may have a different number of off-diagonal points. Additionally, the bottleneck distance will be infinite if and only if two persistence diagrams have a different number of points at infinity.

Following the approach in \cite{CM2018-Dowker}, we will need the next lemma to prove stability. 
Two simplicial maps 
$f,g:\Sigma\rightarrow\Sigma'$ 
are \emph{contiguous} if for any simplex $\sigma\in\Sigma$, the union 
$f(\sigma)\cup g(\sigma)$ 
is a simplex in $\Sigma'$.

\begin{lemma}[{Stability Lemma, \cite[Lemma~8]{CM2018-Dowker}}]\label{lem:stability}
    Let $F,G$ be two filtered simplicial complexes given by
    \[\{s_{\delta,\delta'} : F_\delta\hookrightarrow F_{\delta'}\}_{\delta\leq\delta'\in\R} \quad\text{and}\quad
    \{t_{\delta,\delta'} : G_\delta\hookrightarrow G_{\delta'}\}_{\delta\leq\delta'\in\R}.\]
    Suppose there is $\eta\geq0$ for which there exist families of simplicial
    maps
    $\{\phi_\delta:F_\delta\rightarrow G_{\delta+\eta}\}_{\delta\in\R}$ and
    $\{\psi_\delta:G_\delta\rightarrow F_{\delta+\eta}\}_{\delta\in\R}$
    such that all the following pairs of maps are contiguous:

    \begin{enumerate}[(i)]
        \item $\phi_{\delta'}\circ s_{\delta,\delta'}$
            and $t_{\delta+\eta,\delta'+\eta}\circ
            \phi_\delta$,\label{item:stability-phimaps}
        \item $\psi_{\delta'}\circ t_{\delta,\delta'}$
            and $s_{\delta+\eta,\delta'+\eta}\circ \psi_\delta$,
            \label{item:stability-psimaps}
        \item $s_{\delta,\delta+2\eta}$
            and $\psi_{\delta+\eta}\circ\phi_{\delta}$,
        \item $t_{\delta,\delta+2\eta}$\label{item:stability-psiphi}
            and
            $\phi_{\delta+\eta}\circ\psi_{\delta}$.\label{item:stability-phipsi}
    \end{enumerate}

    Also, for each $k\in\Z_+$, let $\mathcal{H}_k(F)$ and $\mathcal{H}_k(G)$ denote the
    $k$-dimensional persistence vector spaces associated to $F$ and $G$,
    respectively. Then,
    for each $k\in\Z_+$
    \[d_B( \emph{Dgm}_k(\mathcal{H}_k(F)),\emph{Dgm}_k(\mathcal{H}_k(G)) ) \leq \eta.\]
\end{lemma}

\subsection{Graphs and Networks}\label{sec:graphs-networks}

In this section we give the relevant definitions and notation, largely following the setup of \cite{Chowdhury2023}, to define the network distance. 
A \emph{weighted directed graph (digraph)} is a triple~$D=(V,E,w)$, where $V$
is a finite set of vertices,~$E\subseteq V\times V$ is a set of directed
edges,
and there is a weight on the edges
given by~$w:E\rightarrow\R_{\geq 0}$.
We assume that $w(v,v') =0$ if and only if $v=v'$.\footnote{This definition is that of a \textit{dissimilarity network} in \cite{Chowdhury2023}.}
A \emph{walk of length $n$} in $D$ is any sequence of
vertices $\gamma=(v_0,\dots,v_n)$ where~$(v_i,v_{i+1})\in E$ for~$0\leq i\leq
n-1$. Here the vertices can be repeated, and we call the walk a \emph{path} if no
vertices are repeated.
The \emph{weight of the walk} is given by the sum of the weight of the edges in the walk and is denoted as
\[
w(\gamma) = w(v_0,\dots,v_n) := \sum_{i=0}^{n-1}
w(v_i,v_{i+1}).
\]
The weighted digraph $D$ is said to be \emph{strongly connected} if, for any
pair $(u,v)\in V\times V$, there exists a walk in $D$ that starts at $u$ and
ends at $v$.
A complete weighted digraph $(V,V\times V, \omega)$ may sometimes be denoted
simply as $(V,\omega)$.
Note that we use the terms complete weighted digraph and \emph{network} interchangeably, following terminology in \cite{Carlsson2014}.
Let the collection of all networks be denoted $\networks$.

Given a strongly connected weighted digraph $D=(V,E,w)$, the
\emph{shortest distance function}~$\omega_D \colon V \times V \to \R_{\geq 0}$
is given by
\[
\omega_D(u,v)
:= \min\{w(\gamma) : \gamma \text{ is a walk in $D$ from $u$ to $v$}\}.
\]
The complete weighted digraph $(V,\omega_D)$ is called the \emph{shortest-distance
digraph associated to $D$}. We also define a \emph{shortest-distance digraph} in general as a graph that is equal to its associated shortest-distance digraph.

Next, we need to define ways to compare input networks to each other. 
To that end, we start with a careful definition of a metric and the variations we will use throughout this paper. 

\begin{definition}\label{def:metricproperties}
    Let $X$ be a set, and $d: X \times X \to \overline{\R}_{\geq 0}$ any function, 
    where $\overline{\R}_{\geq 0}$ denotes the non-negative extended real line $\R_{\geq 0} \cup \{ \infty\}$. 
    Consider the following properties:
    \begin{itemize}
        \item Finiteness: $\forall x,y \in X$, $d(x,y) < \infty$. 
        \item Identity: $\forall x \in X$, $d(x,x) = 0$. 
        \item Symmetry: $\forall x,y \in X$, $d(x,y) = d(y,x)$. 
        \item Separability: $\forall x,y \in X$, $d(x,y) = 0$ implies $x = y$. 
        \item Subadditivity (Triangle inequality): $\forall x,y,z \in X$, $d(x,y) \leq d(x,z) + d(z,y)$. 
    \end{itemize}
    The function $d$ is called a \textit{metric} if it satisfies all five properties,
    a \textit{semimetric} if it satisfies all but subadditivity, and a \textit{pseudometric} if it satisfies all but separability.
\end{definition}

Now we consider the
$\infty$-network distance $d_\cN$ presented in \cite{Chowdhury2023} generalizing the Gromov--Hausdorff distance between metric spaces \cite{Gromov2007}. 
Although \cite{Chowdhury2023} makes careful treatment of infinite sets $X$, in this paper we will focus on the following finite restriction. 
Let $\cX = (X,\omega_X)$, $\cY = (Y,\omega_Y)$ be complete weighted finite digraphs and let $R$ be any nonempty
relation between~$X$ and $Y$, that is, a nonempty subset  $R \subseteq X\times Y$.
The \emph{distortion} of the relation $R$ is given by
\[
\text{dis}(R) := \max_{(x,y),(x',y')\in R} |\omega_X(x,x')-\omega_Y(y,y')|.
\]
A \emph{correspondence} between $X$ and $Y$ is a relation $R$ between $X$
and~$Y$ such that $\pi_X(R)=X$ and $\pi_Y(R)=Y$, where~$\pi_X$ and~$\pi_Y$ denote
the projections into $X$ and $Y$, respectively.  The collection of all correspondences between $X$ and
$Y$ is denoted $\allcorr(X,Y)$. Now we are ready to define a distance between
networks, which came from some restricted version studied in \cite{Carlsson2014,Chowdhury2015,Chowdhury2016}.

\begin{definition}[\cite{Chowdhury2023}]
\label{def:networkDistance}
    The \emph{$\infty$-network distance}
    $d_\networks : \networks \times \networks \rightarrow\mathbb{R}$ is
    defined by
    \[
    d_\cN(\cX,\cY)
    := \frac12 \min_{R\in\allcorr(X,Y)} \dis(R).
    \]
\end{definition}

The $\infty$-network distance is a metric with our assumptions \cite[Theorem 2.8.1]{Chowdhury2023}. 
However, while it can be defined for more general inputs, the $\infty$-network distance is actually a pseudometric without the positivity and 0-weight self loop assumptions used here. 
The set of those more generalized networks for which $d_\cN(X,Y) = 0$ is completely characterized by the idea of \textit{weakly isomorphic networks}; see \cite[Theorem 2.3.7]{Chowdhury2023} for further details. 

The $\infty$-network distance can be reformulated in a way that is more convenient for showing stability and for computation as follows.
Given $\cX=(X,w_X), \cY=(Y,w_Y) \in \networks$
and a pair of maps~$\phi:X\rightarrow Y$ and $\psi:Y\rightarrow X$,
 the \emph{distortion} and \emph{codistortion} of these maps is
\begin{align*}
    &\dis(\phi) := \max_{x,x'\in X} |\omega_X(x,x')-\omega_Y(\phi(x),\phi(x'))|,\\
    &\codis(\phi,\psi) := \max_{(x,y)\in X\times Y} |\omega_X(x,\psi(y))-\omega_Y(\phi(x),y)|,
\end{align*}
and similarly for $\dis(\psi)$ and $\codis(\psi,\phi)$.
Then, as shown in \cite[Proposition~9]{CM2018-Dowker}, the $\infty$-network distance
can be reformulated as
\[
    d_\cN(\cX,\cY) = \frac12 \min_{\phi,\psi}
    \big\{
        \max \left\{\dis(\phi),
        \dis(\psi), \codis(\phi,\psi), \codis(\psi,\phi)
        \right\}
    \big\}.
\]

\subsection{Dowker Persistence}
\label{ssec:dowker}

In this section, we give the definition for the Dowker filtration following \cite{CM2018-Dowker}. 
Let $(X,\omega_X)\in\cN$ be a network.
Given any $\delta\in\R$, define the following relation on $X$:
$$R_{\delta,X} := \{(x,x') : \omega_X(x,x')\leq\delta\}.$$
Note that $R_{\delta,X}=X\times X$ for some sufficiently large $\delta$, and for any $\delta\leq\delta'$, we have $R_{\delta,X}\subseteq R_{\delta',X}$. 
Using this relation, define the simplicial complex
$$\Dsi_\delta := \{\sigma=[x_0,\dots,x_n] : \text{ there exists } x'\in X \text{ such that } (x_i,x')\in R_{\delta,X} \text{ for each } x_i \}.$$
This is called the \textit{Dowker $\delta$-sink simplicial complex}.

Since $R_{\delta,X}$ is an increasing sequence of sets, it follows that $\Dsi_\delta$ is an increasing sequence of simplicial complexes. In particular, for $\delta\leq\delta'$, there is a natural inclusion map $\Dsi_\delta\hookrightarrow\Dsi_{\delta'}$. The filtration $\{\Dsi_\delta \hookrightarrow \Dsi_{\delta'}\}_{\delta \leq \delta'}$ associated to $X$ is denoted $\Dsi_X$. This is called the \textit{Dowker sink filtration on $X$}.

There is a dual construction, $\Dso_X$, called the \textit{Dowker source filtration}, consisting of the simplicial complexes defined as 
$$\Dso_\delta := \{\sigma=[x_0,\dots,x_n] : \text{ there exists } x'\in X \text{ such that } (x',x_i)\in R_{\delta,X} \text{ for each } x_i \}.$$
The sink and source filtrations are not equal in general, but the persistence diagrams are the same (see \cite[Corollary 20]{CM2018-Dowker}), hence we denote the $k$-dimensional persistence diagram arising from either of these filtrations by $\Dgm_k^\mathfrak{D}(X)$. The Dowker persistence diagram is shown to be stable, as stated in the following result.

\begin{theorem}[{\cite[Proposition 15]{CM2018-Dowker}}]
Let $\cX = (X,\omega_X)$, $\cY = (Y,\omega_Y) \in\networks$. Then 
\[
d_B(\Dgm_k^\mathfrak{D}(\cX),\Dgm_k^\mathfrak{D}(\cY)) \leq 2d_\networks(\cX,\cY).
\]
\end{theorem}

\section{The Walk-Length Filtration}\label{sec:walk-length-filtration}

We now present the new definition of the walk-length filtration and give some properties on its construction using shortest distances,
then we show two versions of stability of walk-length persistence. 
Finally, we give theoretical and practical details on the computation.

\subsection{Definition and Properties}

\begin{definition}[Walk-Length Filtration]\label{def:walk-length}
    Let $D=(V,E,w)$ be a weighted directed graph. For any subset
    of vertices~$\sigma \subseteq V$, define
    \[
    \f_D(\sigma) = \inf\{ w(\gamma) \;:\; \gamma \text{ is a walk in $D$ that contains all vertices in $\sigma$} \}.\]
    Note that, if there is no such walk, $\f_D(\sigma)=\infty$.
    Then, given $\delta \in \R$, we define a simplicial complex associated to
    the length $\delta$ as
    \[
        \Wd = \{\sigma \subseteq V:\; \f(\sigma)\leq\delta \}.
    \]
    The \emph{walk-length filtration of $D$} is the parameterized
    collection of
    simplicial complexes~$\{ \Wd \}_{\delta \in \R}$.
    When necessary, we denote this complex by $\Wd^D$ for clarity.
    We denote as $\Dgm_k^{\text{WL}}(D)$ the $k$-dimensional persistence diagram associated to the corresponding walk-length filtration $\{\Wd^D\}_{\delta\in\R}$.
\end{definition}

To see that the walk-length filtration is indeed a filtration, note
that $\Wd \subseteq \W_{\delta'}$ for $\delta \leq
\delta'$ follows from the definition of $\Wd$.
We next present a lemma providing insights into when the complete complex on $V$
is present in the walk-length filtration.

\begin{lemma}[Realizing Complete Complex]
    Let $D=(V,E,w)$ be a weighted digraph, and $\{\Wd\}_{\delta
    \in \R}$ the associated walk-length filtration.  
    Then, there is a $\delta \in \R$ such that $\Wd$ is the complete simplicial complex on $V$ 
    if and only if $D$ is strongly connected.
\end{lemma}
\begin{proof}
    Let~$\{v_i\}$ be an arbitrary ordering of the vertices in $V$.
    If $D$ is strongly connected, there is a path from $v_i$ to
    $v_{i+1}$ for every $i$.
    Concatenating these together gives a walk that uses all the vertices. Let
    $\delta$ be the length of this walk; then, $\Wd$ is the complete
    simplicial complex on $V$.

    On the other hand, if $D$ is not strongly connected, then there exists
    a pair $u,v\in V$ such that there is no walk from $u$ to $v$. By definition
    of $\f_D$, we have $\f(\{u,v\})=\infty$. Hence, there does not exist~$\delta
    \in \R$ for which $[u,v] \in \Wd$.
\end{proof}

Now we present two lemmas that give us the ability to compute the walk-length filtration in a more simple manner, improving computation time. 
This is done using shortest-distance digraphs; a simple example is shown in Figure~\ref{fig:shortest-dist-graph}.
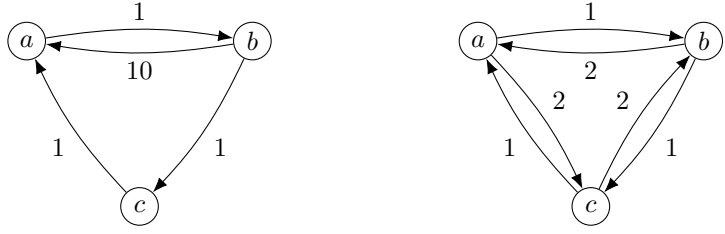
\begin{figure}
\centering
\input{figures/tikz/shortest_dist_graph}
\caption{An example of a shortest-distance digraph (right) obtained from a strongly connected digraph (left).}
\label{fig:shortest-dist-graph}
\end{figure}

\begin{lemma}
\label{lem:WL_same_for_shortest_path}
    Let $D=(V,E,w)$ be a strongly connected weighted digraph, and let
    $\cX=(V,\omega_D)$ be the shortest-distance digraph associated to $D$. Then,
    the walk-length filtrations for $D$ and~$\cX$ are the same.
\end{lemma}

\begin{proof}
    We must show that, for all $\delta \in \R$, the two corresponding simplicial complexes
    associated to $\delta$ are the same, that is, $\Wd^D = \Wd^{\cX}$.
    From the definition, for any directed edge $(a,b)\in V\times V$ there is a
    nontrivial shortest walk~$\gamma = (v_0,\dots,v_m)$ in $\cX$, such that $v_0=a$, $v_m=b$, and
    $w(\gamma)=\omega_D(a,b)$. By concatenation, any walk in $\cX$ can be
    associated to some walk in $D$ that has the same total weight and contains
    all the vertices. This implies that $\f_{\cX}(\sigma)=\f_D(\sigma)$ for any
    subset of vertices $\sigma$; thus, the filtrations are equal.
\end{proof}

\begin{lemma}
    Let $\cX=(V,\omega)$ be a shortest-distance digraph. Given an $n$-simplex $\sigma=(x_0,\dots,x_n)\subset V$, $n\geq1$, in the walk-length filtration $\{\Wd^\cX\}_{\delta\in\R}$, the filtration value of $\sigma$ is given by
    \[
    \f_\cX(\sigma) = \min_{\rho\in S_{n+1}} \sum_{i=0}^{n-1} \omega(x_{\rho(i)},x_{\rho(i+1)})
    = \min_{\rho\in S_{n+1}} \omega(x_{\rho(0)},\dots,x_{\rho(n)}),
    \]
    where $S_{n+1}$ denotes the set of all permutations of $\{0,1,\dots,n\}.$ More concisely, the minimal walk containing $\sigma$ can be given as an ordering of its vertices with no additional vertices needed.
\end{lemma}
\begin{proof}
    Let $\gamma = (v_0,\dots,v_m)$ be a walk in $\cX$ which contains all $x_i\in\sigma$ that is minimal in the sense that $\f_\cX(\sigma) = \omega(\gamma)$. 
    Suppose that $v_i$ is not contained in $\sigma$ for some $0\leq i\leq m$. 
    Firstly, $v_i$ cannot be the start or end of $\gamma$ since then either $(v_1,\dots,v_m)$ or $(v_0,\dots,v_{m-1})$ would still contain $\sigma$ and have lower weight; thus $0<i<m$. 
    Now, given that $\cX$ is a shortest-distance digraph, we know that 
    $\omega(v_{i-1},v_{i+1}) \leq \omega(v_{i-1},v_i,v_{i+1})$, meaning that deleting the vertex $v_i$ will not affect the minimality of the walk. 
    Indeed, if $\gamma'$ denotes the walk obtained after deleting all the $v_i$ that are not in $\sigma$, then 
    $\omega(\gamma') \leq \omega(\gamma)$, and since $\gamma$ has a minimal weight, we have the equality.
    We conclude that $\f_\cX(\sigma) = \omega(\gamma')$, where $\gamma'$ is just a walk through the vertices in $\sigma$ in a certain order.
\end{proof}

\subsection{Stability}
Stability is the notion that whatever representation of our data we have---in this case, the resulting persistence diagram---a function of the distance between the input networks is an upper bound for the distance between the representations.
The standard metric to use in the persistence for directed graphs literature is
to use the $\infty$-network distance $d_\cN$ (Definition~\ref{def:networkDistance}) in these
stability statements \cite{CM2018-Dowker,CM2018-PPH,Turner2019a}.
We first show a negative result: The walk-length filtration is not stable under $d_\cN$.

Consider the example graphs $D$ and $D_\e$ shown in \figref{nonstable-triangles}.
\begin{figure}
    \input{figures/tikz/unstability_example}
    \caption{A case where walk-length construction is not stable under an $\ell_\infty$-distance.}
    \label{fig:nonstable-triangles}
\end{figure}
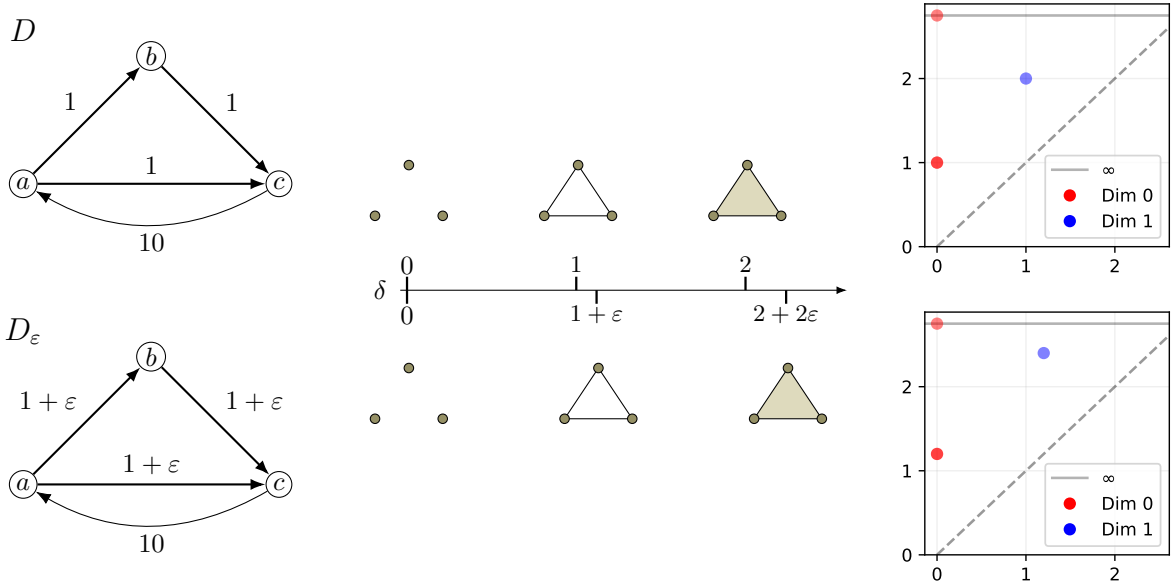
Let $X = \{a,b,c\}$, and let~$\cX = (X,\omega)$ and $\cX_\e = (X,\omega_\e)$ be
the corresponding shortest-distance digraphs of $D$ and~$D_\e$, respectively.
It can be checked that the $\infty$-network distance
$d_\cN(\cX,\cX_\e)$ for this example is given by
the maps $\phi$ and $\psi$ which are both the identity map on $X$.
Then $\dis(\phi) = \dis(\psi) = \codis(\phi,\psi) = \codis(\psi,\phi) =\e$ so
$d_{\cN}(\cX,\cX_\e)=\e/2$.

On the other hand, %
the bottleneck distance between the 1-dimensional diagrams is
\begin{equation*}
    d_B(\DgmWL_1(\cX),\DgmWL_1(\cX_\e)) = \max \left\{2\e, \frac{1+\e}{2} \right\}
\end{equation*}
given by either matching the off-diagonal points to each other, or matching them each to the diagonal.
Regardless of which is larger, we can see that $d_B >\e/2 = d_{\cN}$.
Furthermore, note that the multiplicative factor for $2\epsilon$ in the bottleneck distance %
is directly related to the number of vertices in the cycle representing this one-dimensional
persistence point, while $d_\cN(\cX,\cX_\e)$ is independent of it, so no multiple of $d_\cN$ can bound $d_B$, meaning that walk-length
filtration defies stability with the $\infty$-network distance.

In essence, this issue arises because our filtration function is a sum in the
path length but the $\infty$-network distance is a maximum, hence we should be using a summation version of this distance instead.
We will show two versions of stability: the first one involving a factor of size of the networks (Theorem~\ref{thm:d1-stability-M}), then a second bound free of this factor but restricted to comparison between networks with the same size (Theorem~\ref{thm:d1-stability}).
These two versions use slightly different definitions of network distance, which we give next.

We begin with the summation version of distortion and the corresponding network distance analogue, as seen previously in Section~\ref{sec:graphs-networks}. 
\begin{definition}
    The $\ell_1$-distortion of a relation $R\subset X\times Y$ is given by
    \[\dis^1(R) := \sum_{(x,y),(x',y')\in R} |\omega_X(x,x')-\omega_Y(y,y')|.\]
    Then the network $\ell_1$-distance
    $d_\networks^1:\networks \times \networks \rightarrow \R$ is
    defined as
    \[d_\networks^1(\cX,\cY) := \frac12 \min_{R\in\mathscr{C}(X,Y)} \dis^1(R).\]
\end{definition}

We first verify that $d_\cN^1$ is a semimetric and then we show how $d_\cN^1$ does not satisfy the triangle inequality.

\begin{proposition}
    $d_\cN^1: \cN\times\cN \rightarrow \R$ is a semimetric.
\end{proposition}
\begin{proof}
    Proving the first three properties in Definition~\ref{def:metricproperties} is straightforward; we give the proof for separability.
    Let $\cX,\cY\in\cN$ be two networks.
    If $d_\cN^1(\cX,\cY)=0$, then there is $R\in\mathscr{C}(X,Y)$ such that $\dis^1(R)=0$, which implies that 
    $\omega_X(x,x')=\omega_Y(y,y')$ for all $(x,y),(x',y')\in R$. In particular, setting $x=x'$ implies $y=y'$ because $\omega(a,b)=0$ if and only if $a=b$ according to our definition. We can conclude that $R$ gives us a bijection between $X$ and $Y$ that preserves all weights and thus the networks $\cX$ and $\cY$ are isomorphic.
\end{proof}

To see why $d_\cN^1$ does not satisfy the triangle inequality, consider the example of Figure~\ref{fig:triangleineq-counterexample}. 
For comparing $\cX$ and $\cY$, combinatorial checking yields that the best relation is $\{(x_1,y_1), (x_2,y_2)\} \subseteq \cX \times \cY$, which has distortion 5.
For comparing $\cX$ and $\cZ$, we use the relation with all all pairs having the same subscript; $R = \{(x_1, z_1), (x_1, z_1'), (x_2,z_2) \} \subseteq \cX \times \cZ$ resulting in a distortion of $10.2$. 
A similar relation for $\cY$ and $\cZ$, $\{(y_1, z_1), (y_1, z_1'), (y_2,z_2) \} \subseteq \cY \times \cZ$, gives a distortion of $0.2$. 
The result is that 
$d_\cN^1(\cX,\cY) + d_\cN^1(\cY,\cZ) = 5 + 0.2 = 5.2$ but this is smaller than $d_\cN^1(\cX,\cZ)=10.2$, thus violating the triangle inequality.

\begin{figure}
    \centering
    \input{figures/tikz/triangleineq_counterexample}
    \caption{Counterexample to triangle inequality for $d_\cN^1$.}
    \label{fig:triangleineq-counterexample}
\end{figure}
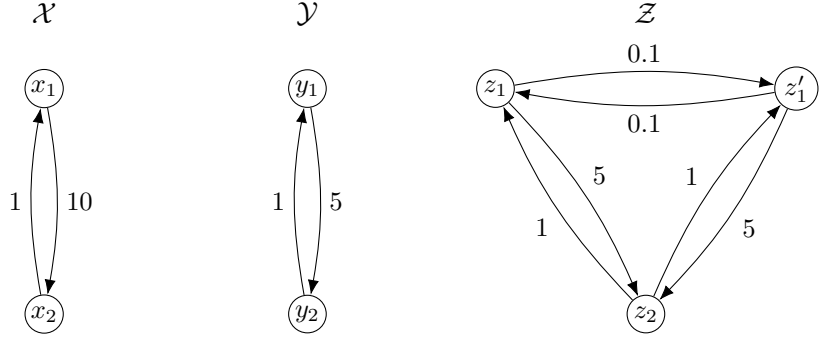

As noted earlier, the  $\infty$-network distance $d_\networks$  can be formulated in two equivalent ways: either using the distortion of a relation or the distortion of a pair of maps. 
However, the parallel formulations are not equivalent in the case of $d_\networks^1$.
We first give an example where the relation version and map version of the distances are different, and then we show the general result.

\begin{definition}
For $\cX,\cY\in \networks$ and any two maps
$\phi:X\rightarrow Y$ and
$\psi:Y:\rightarrow X$
on their sets of vertices,
the \emph{$\ell_1$-distortion} and \emph{$\ell_1$-codistortion} terms are defined respectively as
\begin{align*}
\dis^1(\phi)
    &:= \sum_{x,x'\in X} |\omega_X(x,x')-\omega_Y(\phi(x),\phi(x'))|,\\
\codis^1(\phi,\psi)
    &:= \sum_{(x,y)\in X\times Y} |\omega_X(x,\psi(y))-\omega_Y(\phi(x),y)|,
\end{align*}
with analogous definitions for $\dis^1(\psi)$ and $\codis^1(\psi,\phi)$.
Then, define
\begin{equation}\label{eqn:dn-map}
d_\cN^{1,\mathrm{map}}(X,Y)
    := \frac12 \min_{\substack{\phi,\psi }}
    \big\{
    \max\{\dis^1(\phi), \dis^1(\psi), \codis^1(\phi,\psi), \codis^1(\psi,\phi)\} \big\}.
\end{equation}
\end{definition}

\begin{proposition}\label{prop:reformulation-ineq}
    Let $\cX$ and $\cY$ be two networks. Then,
    \[
    d_\networks^{1,\mathrm{map}}(\cX,\cY)
    \leq
    d_\networks^{1}(\cX,\cY)
    .\]
\end{proposition}

\begin{proof}
    Let $\cX = (X,\omega_X)$ and $\cY = (Y,\omega_Y)$.
    From any correspondence $R\in\allcorr(X,Y)$ we can define a (non-unique)
    pair of maps $\phi:X\rightarrow Y$ and $\psi:Y\rightarrow X$ as follows.
    Since the
    projections are surjective, for each $x\in X$ take any
    $(x,y)\in\pi_X^{-1}(\{x\})$ and define $\phi(x)=y$. The map $\psi$ is
    defined in the same way.

    Then, the set of all pairs of the form $(x,\phi(x))$ is a subset of
    $R$, so that $\dis^1(\phi)$ is a sum over a subset of the summands in $\dis^1(R)$, and as all values are positive we have $\dis^1(\phi)\leq\dis^1(R)$. Similarly
    $\dis^1(\psi)\leq\dis^1(R)$.
    Lastly, we have $\codis^1(\phi,\psi)\leq\dis^1(R)$ because codistortion
    is the sum over the subset of $R\times R$ given by all pairs of the form
    $(x,\phi(x)),(\psi(y),y)\in R$.
    The same is true for $\codis^1(\psi,\phi)$.
    We conclude that for any $R\in\allcorr$ there exist $\phi,\psi$ such that
    $\dis^1(R) \geq \max\{\dis^1(\phi), \dis^1(\psi),
    \codis^1(\phi,\psi), \codis^1(\psi,\phi)\}$
    and the result follows.
\end{proof}

To see that the inequality in Proposition~\ref{prop:reformulation-ineq} can be strict, consider the networks $\cX$ and $\cY$ in Figure~\ref{fig:strict-ineq}.
To obtain a minimal distortion, $a$ and $b$ in $X$ must be paired with $\alpha$, but also
$\beta$ and $\gamma$ must be paired with $c$, that is, we use the correspondence
$R=\{(a,\alpha),(b,\alpha),(c,\beta),(c,\gamma)\}$ and obtain $\dis^1(R) = 4.4$. Any other pairing would increase distortion due to the higher weights in the edges going up in the figure.
On the other hand,  consider the pair of maps
$\phi: \{a\mapsto\alpha,
b\mapsto\alpha,c\mapsto\gamma\}$
and
$\psi:\{\alpha\mapsto a, \beta\mapsto c, \gamma\mapsto c\}$.
Then
$\dis^1(\phi)=\dis^1(\psi)=2.2$ and
$\codis^1(\phi,\psi)=\codis^1(\psi,\phi)=2.7$.
Thus, we see that $d_\cN^{1,\mathrm{map}} = 2.7 < 4.4 = d_\cN^1$ for this example, so we do not achieve equality of the formulations as in the $\ell_\infty$ version,
in which here all distortions and codistortions
for $R,\phi,\psi$ would be equal to $0.5$.

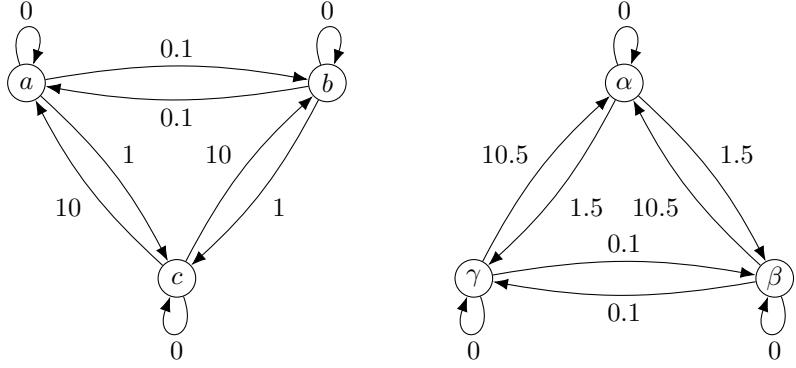
\begin{figure}
    \centering
    \begin{tikzpicture}[>={Latex[length=2mm]}]
        \node[minimum size=0.5cm,draw,circle,inner sep=1pt] (a) at (0,2) {$a$};
        \node[minimum size=0.5cm,draw,circle,inner sep=1pt] (b) at (4,2) {$b$};
        \node[minimum size=0.5cm,draw,circle,inner sep=1pt] (c) at (2,-0.6) {$c$};
        \draw[->] (a) edge["0.1", bend left=10] (b);
        \draw[->] (b) edge["0.1", bend left=10] (a);
        \draw[->] (b) edge["1", bend left=10] (c);
        \draw[->] (c) edge["10", bend left=10] (b);
        \draw[->] (c) edge["10", bend left=10] (a);
        \draw[->] (a) edge["1", bend left=10] (c);
        \draw[->] (a) edge["0",in=70, out=110,looseness=10] (a);
        \draw[->] (b) edge["0",in=70, out=110,looseness=10] (b);
        \draw[->] (c) edge["0",in=-110, out=-70,looseness=10] (c);
    \end{tikzpicture}
    \hspace{30pt}
    \begin{tikzpicture}[>={Latex[length=2mm]}]
        \node[minimum size=0.5cm,draw,circle,inner sep=1pt] (a) at (2,2) {$\alpha$};
        \node[minimum size=0.5cm,draw,circle,inner sep=1pt] (b) at (4,-0.6) {$\beta$};
        \node[minimum size=0.5cm,draw,circle,inner sep=1pt] (c) at (0,-0.6) {$\gamma$};
        \draw[->] (a) edge["1.5", bend left=10] (b);
        \draw[->] (b) edge["10.5", bend left=10] (a);
        \draw[->] (b) edge["0.1", bend left=10] (c);
        \draw[->] (c) edge["0.1", bend left=10] (b);
        \draw[->] (c) edge["10.5", bend left=10] (a);
        \draw[->] (a) edge["1.5", bend left=10] (c);
        \draw[->] (a) edge["0",in=70, out=110,looseness=10] (a);
        \draw[->] (b) edge["0",in=-110, out=-70,looseness=10] (b);
        \draw[->] (c) edge["0",in=-110, out=-70,looseness=10] (c);
    \end{tikzpicture}
    \caption{Networks $\cX$ (left) and $\cY$ (right) showing the inequality of Proposition~\ref{prop:reformulation-ineq} can be strict. }
    \label{fig:strict-ineq}
\end{figure}

We now present a first stability bound for $d_\networks^{1,\mathrm{map}}$, which implies the result for
the $\ell_1$-distance $d_\networks^1$.

\begin{theorem}\label{thm:d1-stability-M}
    Let $\cX$ and $\cY$ be two complete weighted digraphs. 
    If $M=\max\{|X|,|Y|\}$, then, for any $k\in\Z_+$,
    \[
    d_B( \DgmWL_k(\cX),\DgmWL_k(\cY) )
    \leq
    2M\, d_{\cN}^{1,\mathrm{map}}(\cX,\cY)
    \leq
    2M\, d_{\cN}^1(\cX,\cY).
    \]
\end{theorem}

We note that the proof method for this stability theorem closely follows that of \cite[Prop.~15]{CM2018-Dowker}. 
\begin{proof}
    As usual, denote $\cX = (X,\omega_X)$ and $\cY = (Y,\omega_Y)$.
    Let $\eta=2M\,d_{\networks}^{1,\mathrm{map}}(X,Y)$ and $\phi,\psi$ be a pair of maps that achieve
    the minimum in Equation~\ref{eqn:dn-map}, so that
    \[
        \eta = M\cdot \max\big\{
            \dis^1(\phi), \dis^1(\psi),
            \codis^1(\phi,\psi), \codis^1(\psi,\phi)
            \big\}.
    \]
    First we show that $\phi$ and $\psi$ induce families of simplicial maps on the walk-length filtrations of $\cX$ and $\cY$ fitting the assumptions of \lemref{stability}.
    Fix $\delta$ and take a simplex
    $\sigma = [v_0,\dots,v_m] \in K_\delta^\cX$.
    We want to show that
    $\phi(\sigma)=[\phi(v_0),\dots,\phi(v_m)]$
    is a simplex in
    $K_{\delta+\eta}^\cY$.
    By definition, there is a minimal walk
    $\gamma=(u_0,\dots,u_n)$
    in $\cX$, with $m\leq n$, containing all the $v_i$ in $\sigma$ in some order
    such that the weight of $\gamma$ is equal to the filtration value of $\sigma$, 
    that is,
    $w(\gamma) = \f_\cX(\sigma) \leq \delta$.
    Note that there is the possibility that the walk $\gamma$ contains repeated edges; 
    although, from the assumption of minimality, we can be sure that an edge will not be repeated more than $m+1$ times, which is the size of $\sigma$, 
    and $m+1\leq |X|\leq M$ as there are no repeated vertices in $\sigma$. 
    When we sum the differences of edge weights along $\gamma$ and $\phi(\gamma)$, we obtain a distortion-like sum over a subset of vertices where a term can be repeated at most $M$ times, therefore we have the bound
    \[
        \sum_{i=0}^{n-1} | \omega_X(u_i,u_{i+1}) - \omega_Y(\phi(u_i),\phi(u_{i+1}))|
        \leq M\cdot \dis^1(\phi) \leq\eta.
    \]
    The walk $\phi(\gamma)$ contains the vertices in $\phi(\sigma)$,
    so $\f_\cY(\phi(\sigma)) \leq \omega_Y(\phi(\gamma))$. Then using a reverse triangle inequality in each summand we have
    \begin{align*}
        \f_\cY(\phi(\sigma)) &\leq \sum_{i=0}^{n-1} \omega_Y(\phi(u_i),\phi(u_{i+1}))
        \leq \sum_{i=0}^{n-1} \omega_X(u_i,u_{i+1}) \;+\eta
        \;=\; \f_\cX(\sigma) +\eta
        \;\leq\; \delta + \eta,
    \end{align*}
    which means that $\phi(\sigma)\in K_{\delta+\eta}^\cY$.
    Doing a similar check for $\psi$,
    we obtain the collection of simplicial maps
    $\phi:K_\delta^\cX\rightarrow K_{\delta+\eta}^\cY$
    and $\psi:K_\delta^\cY\rightarrow K_{\delta+\eta}^\cX$
    for any $\delta\in\R$.

    Now, our goal is to show that $\eta$, $\phi$, and $\psi$ satisfy the conditions in
    \lemref{stability}, so we want to prove the contiguity between the four pairs of maps.
    Let $s$ and $t$ denote the inclusions in the walk-length filtrations as
    $s_{\delta,\delta'}:K_\delta^\cX\hookrightarrow K_{\delta'}^\cX$ and
    $t_{\delta,\delta'}:K_\delta^\cY\hookrightarrow K_{\delta'}^\cY$
    for all $\delta\leq\delta'\in\R$.

    For pair {(i)} in \lemref{stability}, given that $s$ and $t$ are inclusions and
    $\phi_\delta=\phi$ is not dependent on $\delta$,
    $\phi(s(\sigma)) = t(\phi(\sigma))\in K_{\delta'+\eta}^\cY$
    for any simplex
    $\sigma\in K_\delta^\cX$
    and $\delta\leq\delta'$,
    hence the pair is contiguous.
    The same argument works in pair {(ii)}.

    For pair {(iii)}, note that 
    \begin{align*}
        \qquad \sum_{i=0}^{n-1}
        |\omega_X(u_i,u_{i+1}) \;-\; & \omega_X(\psi(\phi(u_i)),\psi(\phi(u_{i+1})))|  \\
        &\leq \sum_{i=0}^{n-1}
        |\omega_X(u_i,u_{i+1}) - \omega_Y(\phi(u_i),\phi(u_{i+1}))| \\
        &\quad\; + \sum_{i=0}^{n-1}
        |\omega_Y(\phi(u_i),\phi(u_{i+1})) - \omega_X(\psi(\phi(u_i)),\psi(\phi(u_{i+1})))| \\
        &\leq M\cdot \dis^1(\phi) +M\cdot \dis^1(\psi)
        \leq 2\eta,
    \end{align*}
    and then a similar application of a triangle inequality as before gives 
    $\f_\cX(\psi(\phi(\sigma))) \leq \delta +2\eta$.
    Writing $\psi(\phi(\sigma))=\sigma'$, this means that $\sigma'$ is a simplex in
    $K_{\delta+2\eta}^\cX$.
    Next we need to show that
    $\sigma\cup\sigma'$
    is also a simplex in
    $K_{\delta+2\eta}^\cX$.

    Denote $\{u_0',\dots,u_n'\}=\gamma'=\psi(\phi(\gamma))$.
    Taking all pairs $(u_i,\phi(u_i))$ in the codistortions, 
    and applying the same principle for bounding possible repetitions, 
    we have
    \begin{align*}
        \sum_{i=0}^{n-1} \omega_X(u_i,u_i')
        &= \sum_{i=0}^{n-1} |\omega_X(u_i,\psi(\phi(u_i)))-\omega_Y(\phi(u_i),\phi(u_i))| \;\leq\; M\cdot \codis^1(\phi,\psi) \;\leq\; \eta,\\
        \sum_{i=0}^{n-1} \omega_X(u_i',u_i)
        &= \sum_{i=0}^{n-1} |\omega_Y(\phi(u_i),\phi(u_i))-\omega_X(\psi(\phi(u_i)),u_i)| \leq M\cdot \codis^1(\psi,\phi) \leq \eta.
    \end{align*}
    To determine $\f_\cX(\sigma\cup\sigma')$, we need a walk that contains all $u_i$
    and $u_i'$.
    Consider the walk that visits the vertices in the order
    $\{u_0,u_0',u_0,u_1,u_1',u_1,u_2,\cdots\}$.
    The total weight of this walk is
    \[
        \sum_{i=0}^{n-1} \omega_X(u_i,u_{i+1})
        + \sum_{i=0}^{n-1} \omega_X(u_i,u_i')
        + \sum_{i=0}^{n-1} \omega_X(u_i',u_i)
        \leq \delta + 2\eta.\]
    Thus $\f_\cX(\sigma\cup\sigma')\leq\delta+2\eta$, i.e.,~%
    $\sigma\cup\psi(\phi(\sigma))\in K_{\delta+2\eta}^\cX$. The same reasoning can
    be applied to $\phi\circ\psi$ to show the contiguity of the pair {(iv)}.
    The result now follows from \lemref{stability} and Proposition~\ref{prop:reformulation-ineq}.
\end{proof}

The bound in Theorem~\ref{thm:d1-stability-M} is not ideal, since having the number of vertices $M$ as a multiplicative factor implies that even with a small perturbation in the networks, that is, a small network distance, the bound can only guarantee a small perturbation in the persistence diagrams proportional to the size of the networks.

Looking at the proof, we can avoid the factor of size if we avoid any repeated terms in the distortion, and the only way to ensure this is when the maps $\phi,\psi$ are injective.
Hence the next stability result in Theorem~\ref{thm:d1-stability} will be restricted to a distance using bijections between networks with the same number of vertices, which does define a metric in this case.
Note that this parallels the similar setup in the $\infty$ setting given in \cite[Definition~2.7.1]{Chowdhury2023} to handle issues with pseudometric properties; in this paper, we will use a parallel idea to provide better stability bounds. 
We define the distance setup with both relations and maps as before, but show that unlike $d_\cN^1$, the constructions are the same in this setting.

\begin{definition}
Denote $\networks_M$ as the collection of networks with $M$ vertices 
and $\mathscr{B}(X,Y)\subset\mathscr{C}(X,Y)$ as the set of matchings (bijections) between $X$ and $Y$.
Define $d_{\cN,\text{bij}}^1:\networks_M \times \networks_M \rightarrow \R$ as
\[d_{\cN,\text{bij}}^1(\cX,\cY) := \frac12 \min_{R\in\mathscr{B}(X,Y)} \dis^1(R).\]
Then, to match this approach, we must define the map version using bijections that are inverses:
\begin{equation}\label{eqn:dn-map-bij}
d_{\cN,\text{bij}}^{1,\mathrm{map}}(\cX,\cY)
    := \frac12 \min_{\substack{\phi,\psi \\ \text{bijective} \\ \phi = \psi^{-1}}}
    \big\{
    \max\{\dis^1(\phi), \dis^1(\psi), \codis^1(\phi,\psi), \codis^1(\psi,\phi)\} \big\}.
\end{equation}
\end{definition}

First, we see that unlike $d_\cN^1$, the bijection version of the 1-network distance $d_{\cN,\text{bij}}^1$ is a metric on $\cN_M$.

\begin{proposition}
    $\dNbij: \cN_M\times\cN_M \rightarrow \R$ is a metric for any $M$.
\end{proposition}
\begin{proof}
    The first four properties in Definition~\ref{def:metricproperties} are obtained similarly as for $d_\cN^1$.
    We give the proof for subadditivity. For any networks $\cX,\cY,\cZ\in \networks$, we want to check the triangle inequality
    \[\dNbij(\cX,\cZ) \leq \dNbij(\cX,\cY) + \dNbij(\cY,\cZ).\]
    Let $R_1\subset X\times Y$ and $R_2\subset Y\times Z$ be two matchings. We
    define the matching $R_3\subset X\times Z$ where a pair $(x,z)$ is in $R_3$
    whenever $\pi_Y(\{x\}\times Y)\cap\pi_Y(Y\times\{z\})\neq\emptyset$, that is,
    there is some $y\in Y$ such that $(x,y)\in R_1$ and $(y,z)\in R_2$. This is also a
    matching, since each $x\in X$ is paired with only one $y\in Y$ and
    this $y$ is paired with only one $z\in Z$; the same is true for each $z\in
    Z$.
    Then any summand in $\text{dis}^1(R_3)$ satisfies
    \begin{align*}
        |\omega_X(x,x') -\omega_Z(z,z')|
            &= |\omega_X(x,x') -\omega_Y(y,y') +\omega_Y(y,y') -\omega_Z(z,z')| \\
            &\leq |\omega_X(x,x') -\omega_Y(y,y')| +|\omega_Y(y,y') -\omega_Z(z,z')|
    \end{align*}
    for a unique pair $y,y'\in Y$. Taking the sum over all $(x,z),(x',z')\in X\times Z$, the
    right side would be a pair of sums on $R_1$ and $R_2$ respectively, thus we see
    that $\text{dis}^1(R_3) \leq \text{dis}^1(R_1) + \text{dis}^1(R_2)$. The
    triangle inequality follows from taking the minimum over
    $R_1\in\mathscr{B}(X,Y)$ and $R_2\in\mathscr{B}(Y,Z)$.
\end{proof}

In additional, also unlike the non-bijective case, we see that the map and relation versions of the definitions coincide. 
\begin{lemma}
\label{lem:Bij_Map_relation_same}
We have
$d_{\cN,\mathrm{bij}}^1(\cX,\cY) =
d_{\cN,\mathrm{bij}}^{1,\mathrm{map}}(\cX,\cY)$.
\end{lemma}
\begin{proof}
Let $\phi:X\rightarrow Y$ and $\psi:Y\rightarrow X$ be a pair of bijections that achieve the minimum in 
$d_{\cN,\mathrm{bij}}^{1,\mathrm{map}}(\cX,\cY)$. 
Since these are inverses, it is easy to see that $\dis^1(\phi)=\dis^1(\psi)=\codis^1(\phi,\psi)=\codis^1(\psi,\phi)$, so we can write 
\[d_{\networks,\text{bij}}^{1,\mathrm{map}} = \frac12 \min_{\substack{\phi \\ \text{bijective}}} \dis^1(\phi) = d_{\networks,\text{bij}}^1.\]
\end{proof}

Finally, we can combine these lemmas with a variation of the same proof technique of Theorem~\ref{thm:d1-stability-M} and \cite[Proposition15]{CM2018-Dowker}. 

\begin{theorem}\label{thm:d1-stability}
Let $C$ and $D$ be two strongly connected weighted digraphs with the same number of vertices, and let $\cX$ and $\cY$ be their corresponding shortest-distance digraphs. Then, for any $k\in\Z_+$,
    \[
    d_B( \DgmWL_k(C),\DgmWL_k(D) )
    =
    d_B( \DgmWL_k(\cX),\DgmWL_k(\cY) )
    \leq
    2\, d_{\networks,\mathrm{bij}}^{1,\mathrm{map}}(\cX,\cY)
    =
    2\, d_{\cN, \mathrm{bij}}^1(\cX,\cY).
    \]
\end{theorem}

\begin{proof}
The first equality is from Lemma~\ref{lem:WL_same_for_shortest_path} and the last equality is from Lemma~\ref{lem:Bij_Map_relation_same}, so all that remains is the middle inequality. 
    Denote $\cX = (X,\omega_X)$ and $\cY = (Y,\omega_Y)$ as usual.
    Let $\eta=2d_{\networks}^{1,\mathrm{map}}(X,Y)$ and $\phi,\psi$ be a pair of bijective maps that achieve
    the minimum in \eqnref{dn-map-bij}, so that
    \[
        \eta = \max\big\{
            \dis^1(\phi), \dis^1(\psi),
            \codis^1(\phi,\psi), \codis^1(\psi,\phi)
            \big\}.
    \]
    From before we have that $\phi=\psi^{-1}$ and $\eta=\dis^1(\phi)$.
    First we show that $\phi$ and $\psi$ induce families of simplicial maps on the walk-length filtrations of $\cX$ and $\cY$ as in \lemref{stability}.
    Fix $\delta$ and take a simplex
    $\sigma = [v_0,\dots,v_m] \in K_\delta^\cX$.
    We want to show that
    $\phi(\sigma)=[\phi(v_0),\dots,\phi(v_m)]$
    is a simplex in
    $K_{\delta+\eta}^\cY$.
    By definition, there is a minimal walk
    $\gamma=(u_0,\dots,u_n)$
    in $\cX$, with $m\leq n$, containing all the $v_i$ in $\sigma$ in some order,
    such that the weight of $\gamma$ is equal to
    $w(\gamma) = \f_\cX(\sigma) \leq \delta$.
    Since $\cX$ is a shortest-distance digraph, we can assume that the walk $\gamma$ has no repeated vertices or edges. To see this, note that if a vertex $u_{i_0}$ appears a second time in $\gamma$, then the sequence $(u_{i_0-1},u_{i_0},u_{i_0+1})$ can be replaced by $(u_{i_0-1},u_{i_0+1})$ to obtain a walk that visits $\sigma$ with less or equal weight. 
    Next, when we sum the differences of edge weights along $\gamma$ and along $\phi(\gamma)$, we obtain a distortion-like sum over a subset of vertices in $X$, therefore we have the following:
    \[
        \sum_{i=0}^{n-1} | \omega_X(u_i,u_{i+1}) - \omega_Y(\phi(u_i),\phi(u_{i+1}))|
        = \text{dis}^1(\phi|_\gamma) \leq \text{dis}^1(\phi) = \eta.
    \]
    The walk $\phi(\gamma)$ contains the vertices in $\phi(\sigma)$,
    so $\f_\cY(\phi(\sigma)) \leq \omega_Y(\phi(\gamma)$. Then using a reverse triangle inequality in each summand we have
    \begin{align*}
        \f_\cY(\phi(\sigma)) &\leq \sum_{i=0}^{n-1} \omega_Y(\phi(u_i),\phi(u_{i+1}))
        \leq \sum_{i=0}^{n-1} \omega_X(u_i,u_{i+1}) \;+\eta
        \;=\; \f_\cX(\sigma) +\eta
        \;\leq\; \delta + \eta,
    \end{align*}
    which means that $\phi(\sigma)\in K_{\delta+\eta}^\cY$.
    Doing the same check for $\psi$,
    we obtain the simplicial maps
    $\phi:K_\delta^\cX\rightarrow K_{\delta+\eta}^\cY$
    and $\psi:K_\delta^\cY\rightarrow K_{\delta+\eta}^\cX$
    for any $\delta\in\R$.

    Now, our goal is to show that $\eta$, $\phi$, and $\psi$ satisfy the conditions in
    \lemref{stability}, that is, we want to prove the contiguity between the four pairs of maps.
    Let $s$ and $t$ denote the inclusions in the walk-length filtrations as
    $s_{\delta,\delta'}:K_\delta^\cX\hookrightarrow K_{\delta'}^\cX$ and
    $t_{\delta,\delta'}:K_\delta^\cY\hookrightarrow K_{\delta'}^\cY$
    for all $\delta\leq\delta'\in\R$.

    For pair \textbf{(i)} in \lemref{stability}, given that $s$ and $t$ are inclusions and
    $\phi_\delta=\phi$ is not dependent on $\delta$,
    $\phi(s(\sigma)) = t(\phi(\sigma))\in K_{\delta'+\eta}^\cY$
    for any simplex
    $\sigma\in K_\delta^\cX$
    and $\delta\leq\delta'$.
    Thus the square is commutative and the pair is contiguous.
    The same argument works in pair \textbf{(ii)}.

    For pair \textbf{(iii)}, simply note that 
    $s(\sigma) = \psi(\phi(\sigma)) = \sigma \in K_\delta^\cX\subset K_{\delta+2\eta}^\cX$
    and contiguity follows. The same applies using $\phi\circ\psi$ for pair \textbf{(iv)}.
    The result now follows from \lemref{stability}.
\end{proof}

\subsection{Algorithm}
\label{ssec:algorithm}

In this section, we develop the computational details on the construction of the walk-length filtration.
In theory, the structure of the algorithm follows a dynamic programming approach, where the filtration value of a simplex is computed using the filtration value of its faces.

Let $D=(V,E,w)$ be a strongly connected weighted digraph with $n = |V|$. 
First, we replace this input with the shortest-distance digraph, $\cX=(V,\omega)$. 
This can be done, for example, with the Floyd-Warshall algorithm in $O(n^3)$ time \cite{Erickson2019}. 
Given $\cX$ as the weighted adjacency matrix $A \in \R_{\geq 0}^{n\times n}$, our next job is to determine the corresponding walk-length simplicial filtration, denoted as $\{\Wd\}_{\delta\geq0}$. 
We can compute the filtration value of a $d$-simplex $\sigma = [v_0,\dots,v_d]$ assuming we have already computed the filtration value for all $(d-1)$ dimensional simplices as follows. 
Let $\sigma_i=[v_0,\dots,\widehat{v_i},\dots,v_d]$ be the $i$-th face of $\sigma$ and define 
$\beta_i = (s_{i},\dots,t_{i})$ 
to be a walk for which $\omega(\beta_i) = \f(\sigma_i)$. 
Then it is straightforward to check that $\f(\sigma)$ can be expressed in terms of the $\f(\sigma_i)$, either by adding the missing vertex at the start or end of the walk, that is,
\begin{align}
\label{eqn:wl-filtvalue}
\f(\sigma) 
& = \min_{0\leq i\leq d} \bigg\{
\min\{\; \omega(v_i,s_i) +\f(\sigma_i), \;\;
\f(\sigma_i) +\omega(t_i,v_i) \;\}
\bigg\}.
\end{align}

For any simplicial filtration, if we want to compute $k$-dimensional persistent homology, we only need the data up to $(k+1)$-simplices, that is, we only need the filtration value of subsets of at most $k+2$ vertices.
In particular, for walk-length persistence, all $0$-simplices (vertices) are added at filtration value $0$. After that, the filtration value of a $1$-simplex is simply obtained from the shortest distance between the pair of vertices in either direction,
$ \f([v_0,v_1]) = \min\{\omega(v_0,v_1), \omega(v_1,v_0)\}$.
For a general $d$-dimensional simplex $\sigma$ in Equation~\ref{eqn:wl-filtvalue}, we  need to evaluate $2(d+1)$ possibilities. 
Note that in this equation, not only do we need to know the filtration values of the simplicies, but also the starting and ending vertices of a walk which attains it, so in our algorithm we ensure that we have stored this information.
See Algorithm~\ref{alg:WalkLengthFiltration} for the pseudocode. 

\begin{algorithm}
\caption{Walk-length filtration}\label{alg:WalkLengthFiltration}
\begin{algorithmic}
\Require Directed shortest-path weight matrix $A$ for vertex set $V$ and a fixed $k \geq 0$. 
\Ensure $\f: \sigma \mapsto \f(\sigma)$ for all simplices $\sigma$ of dimension at most $k+1$. 
\State Initialize $\f(v) = 0$ for all $v \in V$. 
\State Set $s_{v} = v$ and $t_v = v$ for all $v \in V$. 
\For  {$d \in 1, \cdots, k+1$}
\For {$\sigma = [v_0,\cdots, v_d]$ }
\State $\f(\sigma) = \min_{i \in \{0, \cdots, d \}}\{A[v_i, s_{\sigma_i}] + \f(\sigma_i), \f(\sigma_i) + A[t_{\sigma_i}, v_i]  \}$
\Comment{$\sigma_i = [v_0, \cdots, \hat{v_i}, \cdots, v_d]$}
\State Set $s_\sigma$ to be either $v_i$ or $s_{\sigma_i}$ based on minimum chosen above.
\State Set $t_\sigma$ to be either $t_{\sigma_i}$ or $v_i$  based on minimum chosen above.
\EndFor
\EndFor
\end{algorithmic}
\end{algorithm}

Finally, we evaluate the running time and space of Algorithm~\ref{alg:WalkLengthFiltration}. 
Assume we are interested in computing the persistence of this filtration up to a fixed dimension $k$, and thus we are computing the filtration values for simplices up to dimension $k+1$. 
Further, in practice we may assume that $k$ is likely very small, at most 2 in many applications. 
For the sake of this analysis, we thus assume $k = o(n)$. 
In the case of a complete simplicial complex on $n$ vertices, there are $\binom{n}{d}$ simplices of dimension $d$. 
Each of these simplices has $d+1$ faces of dimension $d-1$, and thus for every $d$-dimensional simplex $\sigma$, we need to check $2(d+1)$ values in the minimum to determine $\f(\sigma)$.
Once this is determined, finding the relevant $s_\sigma$ and $t_\sigma$ is a single check, resulting in $O(1)$ to evaluate. 
The result is that the running time is $\sum_{d= 0} ^{k+1} \binom{n}{d+1} \cdot 2(d+1)$. 
For space, while determining the filtration values for dimension $d$, we need only store the labels of two vertices for the $d-1$ dimensional vertices. 
This means that we need at most $O(\binom{n}{k}) = O(n^k)$ space which will be reasonable for small $k$. 

We next get a more understandable bound on the running time with respect to $n$ for a fixed $k$. 
Manipulating binomials and then using Lemma \ref{lem:binom_bound} in the appendix gets us that 
\begin{equation*}
    \sum_{d= 0} ^{k+1} \binom{n}{d+1} \cdot 2(d+1)  
    = 2n \sum_{d=0}^{k+1} \binom{n-1}{d}
    \leq \binom{n-1}{k+1} \frac{n(n-k-1)}{n-2k-2}.
\end{equation*}
This implies that for a fixed $k$ which is $o(n)$, we have that the running time is $O(n^{k+2} )$.

\section{Examples}\label{sec:examples}

In this section we give two studies that parallel the experiments done to examine the Dowker filtration \cite{CM2018-Dowker}. 
First, we work with cycle networks, and show that walk-length persistence exhibits a higher sensitivity to non-directed cycles versus directed cycles.
Second, we compute the walk-length filtration for the simulated hippocampal network data of \cite{CM2018-Dowker} and compare the results using the diagrams in a classifier.

\subsection{Cycle Networks}

For $n\geq3$, let $D_n$ be the cycle graph given by a directed cycle with
$n$ vertices $\{x_1,\dots,x_n\}$ where all edges $(x_i,x_{i+1})$ for $1\leq i\leq n-1$ and $(x_n,x_1)$ have weight $1$. 
Note that this graph is strongly connected and define the \emph{cycle network}
$(G_n,\omega_{G_n})$ as the
associated shortest-distance digraph.
The Dowker persistence for this cycle graph was given in  \cite{CM2018-Dowker}.  
\begin{proposition}[{\cite[Theorem~33]{CM2018-Dowker}}]
    \label{prop:cycle-persistence-Dowker}
    Let $G_n$ be a cycle network with $n\geq3$, then
    \[
    \Dgm_1^\mathfrak{D}(G_n) = \big\{(1,\lceil n/2\rceil)\in\R^2 \big\}.
    \]
\end{proposition}

As we show in the next proposition, the walk-length persistence of the cycle network gives the same result.

\begin{proposition}\label{prop:cycle-persistence}
    Let $(G_n,\omega)$ be a cycle network with $n\geq3$. Denote the Dowker sink and walk-length simplicial filtrations as $\Dsi_\delta$ and $\Wd$, respectively. We have $\Dsi_\delta=\Wd$
    for all $\delta\in\R$, so that $\Dgm_k^\mathfrak{D}(G_n) =
    \DgmWL_k(G_n)$ for all $k\in\Z_+$. In particular, by Proposition~\ref{prop:cycle-persistence-Dowker},
    \[
    \DgmWL_1(G_n) = \big\{(1,\lceil n/2\rceil)\in\R^2 \big\}.
    \]
\end{proposition}

\begin{proof}
    We show that Dowker and walk-length filtrations are the same at any step $\delta\in\R$. 
    Note that every shortest walk in $G_n$ comes from a path in the graph $D_n$, which is a cyclic sequence of edges of weight $1$.

    First, take a simplex $\sigma=[v_0,\dots,v_m]\in \Wd$. By definition, there is a path $\gamma$ in
    $D_n$ that includes all the $v_i$ such that $\omega(\gamma)\leq \delta$.
    Take $s\in\{0,\dots,m\}$ such that $v_s$ is the last vertex in $\gamma$, then $\omega(v_i,v_s)\leq \omega(\gamma) \leq\delta$ for all $i$. 
    Thus $v_s$ is a sink for $\sigma$ and $\sigma\in\Dsi_\delta$.

    Now take $\sigma=[v_0,\dots,v_m]\in\Dsi_\delta$, then there is a sink $v$
    such that $\omega(v_i,v)\leq\delta$ for all $i$. 
    Take $i_0\in\{0,\dots, m\}$ such that $\omega(v_{i_0},v)\geq \omega(v_i,v)$ for all $i$, hence the path from $v_{i_0}$ to $v$ with weight
    $\omega(v_{i_0},v)\leq\delta$ contains all the other $v_i$ and $\sigma\in\Wd$.
    Given that the simplicial filtrations are identical, the resulting persistence diagrams are also the same.
\end{proof}

Next, we study what happens when a single edge in this cycle graph $D_n$ is modified to illuminate differences in what the two filtrations can distinguish in this example. 
We define the graph $\widetilde{D}_n$ as follows: Let the weight of
the edge $(x_n,x_1)$ be larger than or equal to $n-2$, then add the edge $(x_1,x_n)$ with weight $1$.
See Figure~\ref{fig:hexagons} for an example of $\widetilde{D}_n$ with $\omega(x_n,x_1) = n-2$. 
In general, this graph is strongly connected and we can define the associated shortest-distance digraph
$(\widetilde{G}_n,\omega_{\widetilde{G}_n})$. The 1-dimensional persistent homology of this network is described in the following proposition.

Note that, when $\omega(x_n,x_1)=n-1$, this modification is similar to the pair swap in \cite[Definition~6]{CM2018-Dowker}.
However, in that setup a single pair of edges in the full network was swapped, 
whereas in our case the swap is done on the graph $D_n$, 
thus the resulting network $\widetilde{G}_n$ differs from $G_n$ in more edge weights than just one pair.

\begin{proposition}\label{prop:semicycle-persistence}
    Let $(\widetilde{G}_n,\omega)$ be as above, then
    \[
    \Dgm_1^\mathfrak{D}(\widetilde{G}_n) = \big\{(1,\lceil(n-1)/2\rceil)\in\R^2 \big\}
    \quad\text{and}\quad
    \DgmWL_1(\widetilde{G}_n) = \big\{(1,n-1)\in\R^2 \big\}.
    \]
\end{proposition}
Note that this proof largely follows the structure of \cite[Theorem~33]{CM2018-Dowker}, stated here as Proposition~\ref{prop:cycle-persistence-Dowker}. 
\begin{proof}
We give a complete proof for walk-length persistence, then we comment on how the Dowker version works using the same approach.
Recall that $\omega(x_n,x_1)\geq n-2$ and $\omega(x_1,x_n)=1$ in the network $\widetilde{G}_n$.
Also note that $\omega(x_i,x_j)=j-i$ for all $1\leq i\leq j\leq n$ with the only exception of $\omega(x_1,x_n)=1$.
For simplicity, we denote $x_{n+1}:=x_1$.

Let $\Wd$ be the walk-length filtration of $(\widetilde{G}_n,\omega)$. 
We denote chain groups as $C_k^\delta:=C_k(\Wd;\K)$ 
where $\K$ is a field, 
boundary maps as 
$\partial_k^\delta:=\partial_k:C_k^\delta\rightarrow C_{k-1}^\delta$,
and homology groups as $H_k^\delta:=H_k(\Wd)=\Ker(\partial_k^\delta)/\Img(\partial_{k+1}^\delta)$. 

The proof is broken down into five steps, where a homology class is identified as the sole generator of all 1-dimensional persistence.
\begin{enumerate}
    \item Define $\alpha_n = [x_1,x_2] + [x_2,x_3] + \dots + [x_n,x_1]$. Note that $[x_i,x_{i+1}]$ enters the filtration at $\delta=1$ for all $1\leq i\leq n$, as well as $[x_n,x_1]$ since there is an edge $(x_1,x_n)$ with weight $1$, so we see that $\alpha_n\in C_1^\delta$ for $\delta\geq 1$. One can check that $\partial_1^\delta(\alpha_n)=0$, so $\alpha_n\in\Ker(\delta_1^\delta)$ and we obtain the homology class $[\alpha_n]_\delta\in H_1^\delta$.
    
    \item We prove that $\alpha_n$ generates $\Ker(\delta_1^\delta)$ for $1\leq \delta < 2$. The only 1-simplices in $\Wd$ are the ones coming from the edges of $\widetilde{D}_n$ since any other pair would require a larger walk. If there is a 1-chain $z := \sum_{i=1}^{n} b_i[x_i,x_{i+1}]$ such that $\partial_1^\delta(z)=0$, one can check that $b_1=\dots=b_n$, hence $z$ is a multiple of $\alpha_n$.
    In terms of vector spaces, in homology, this means span$\{[\alpha_n]_\delta\} = H_1^\delta$.

    \item We prove that $\alpha_n$ generates $\Ker(\delta_1^\delta)$ for $\delta \geq 2$.
    Take any other 1-chain $y = \sum_{i=1}^m a_i\sigma_i$ for some $a_i\in\K$ and 1-simplices $\sigma_i\in\Wd$. 
    Fix $i$ and write $\sigma_i=[x_p,x_q]$.
    We exclude the cases $p+1=q$, or $p=q+1$, or $p=n$ and $q=1$, since then $\sigma_i$ would already be one of the summands in $\alpha_n$.
    We can assume $p<q$ and thus 
    $\f([x_p,x_q]) = \omega(x_p,x_q) =q-p$.
    Also notice that any 2-simplex $[x_j,x_k,x_l]$, with $1\leq j<k<l\leq n$, has filtration value $\f([x_j,x_k,x_l])=l-j$, hence
    \[
    2 = \f([x_{q-2},x_{q-1},x_q]) \leq \f([x_{q-3},x_{q-2},x_q]) 
    \leq \dots \leq 
    \f([x_{p},x_{p+1},x_q]) = q-p = \omega(x_p,x_q) \leq \delta.
    \]
    That is, we have a 2-chain $[x_{q-2},x_{q-1},x_q] +[x_{q-3},x_{q-2},x_q]+ \dots +[x_{p},x_{p+1},x_q] \in C_2^\delta$ for which
    \begin{align*}
    \partial_2^\delta\left( \sum_{r=0}^{q-p-2} [x_{p+r},x_{p+r+1},x_q] \right) 
    & = \sum_{r=0}^{q-p-2} [x_{p+r+1},x_q] -[x_{p+r},x_q] +[x_{p+r},x_{p+r+1}] \\
    & = \sum_{s=p}^{q-2} [x_{s+1},x_q] -[x_{s},x_q] 
    + \sum_{s=p}^{q-2}[x_{s},x_{s+1}] \\
    & = \sum_{s=p}^{q-1}[x_{s},x_{s+1}] \;-[x_p,x_q] = \sum_{s=p}^{q-1}[x_{s},x_{s+1}] \;-\sigma_i.
    \end{align*}
    This means that $\sum_{s=p}^{q-1}[x_{s},x_{s+1}] -\sigma_i \in \Img(\partial_2^\delta)$. 
    After doing the same for each $\sigma_i$, we can obtain a linear combination $z := \sum_{i=1}^{n} b_i[x_i,x_{i+1}]$ for some coefficients $b_i\in\K$ such that
    $z-y \in \Img(\partial_2^\delta)$. Furthermore, we can see that 
    $z \in \Ker(\partial_1^\delta)$,
    and so in homology, similarly as before, we obtain $[y]_\delta = [z]_\delta \in \text{span}\{[\alpha_n]_\delta\} = H_1^\delta$.
    
    \item We show that $[\alpha_n]_\delta$ is not trivial for $1\leq \delta<n-1$. Suppose it is trivial, that is, there is a 2-chain $\beta$ such that $\partial_2^\delta(\beta) = \alpha_n$. This requires that $\beta$ contains the simplex $[x_1,x_j,x_n]$ for some $1<j<n$, but the shortest walk in $\widetilde{D}_n$ that contains these three vertices is the path going all around from $x_1$ to $x_n$, which has a weight of $n-1$.
    
    \item Define $\beta_n = \sum_{i=2}^{n-1} [x_1,x_i,x_{i+1}]$, then $\beta_n\in C_2^\delta$ for $\delta\geq n-1$,
    since the maximum filtration value for the summands is $\f([x_1,x_{n-1},x_n])=n-1$.\footnote{This is the reason to define $\widetilde{D}_n$ with $\omega(x_n,x_1) \geq n-2$.}
    Note that 
    \begin{align*}
    \partial_2^\delta\bigg( \sum_{i=2}^{n-1} [x_1,x_i,x_{i+1}] \bigg) 
    & = \sum_{i=2}^{n-1} [x_i,x_{i+1}] -[x_1,x_{i+1}] +[x_1,x_i] \\
    & = \sum_{i=1}^{n-1} [x_i,x_{i+1}] \;-[x_1,x_n] 
    = \alpha_n \in \Img(\partial_2^\delta),
    \end{align*}
    thus $[\alpha_n]_\delta$ is trivial for $\delta\geq n-1$.
\end{enumerate}
We conclude that all 1-dimensional homology classes in walk-length persistence are generated by $[\alpha_n]_\delta$, which is born at $\delta=1$ and dies at $\delta=n-1$.

For Dowker persistence, we can follow the same approach, where the main difference is the entering of the $2$-chain $\beta_n$. 
Each simplex $[x_1,x_i,x_{i+1}]$, for $2\leq i\leq n-1$, enters the Dowker sink filtration $\Dsi_\delta$ at $\delta=\min\{ i, n-i\}$, where these two values come from using either $x_{i+1}$ or $x_n$ as the sink.
We can check that 
\[
\max_{2\leq i\leq n-1} \min\{ i, n-i\} = \left\lceil\frac{n-1}{2}\right\rceil,
\]
thus $\beta_n \in C_2^\delta$ only for $\delta\geq \lceil(n-1)/2\rceil$. As before, $\partial_2^\delta(\beta_n) = \alpha_n$ and so the homology class $[\alpha_n]_\delta$ dies at $\delta = \lceil(n-1)/2\rceil$ in Dowker persistence.
\end{proof}

We now mention two interesting examples, then a general conclusion from our analysis with cycle networks.

\begin{example}
In the case of three vertices, we have $\Dgm_1^\mathfrak{D}(G_3) =
\DgmWL_1(G_3) = \{(1,2)\}$. After the modification given above, we have $\Dgm_1^\mathfrak{D}(\widetilde{G}_3)=\emptyset$ but for walk-length it is unchanged, 
which means that walk-length persistence cannot differentiate between cyclic and non-cyclic triangles although Dowker can.
\end{example}

\begin{example}
For the cycle graph $G_n$ and modified graph $\widetilde{G}_n$ with $n\geq4$ even, by Propositions~\ref{prop:cycle-persistence} and \ref{prop:semicycle-persistence}, Dowker persistence does not differentiate between $G_n$ and
$\widetilde{G}_n$ in dimension 1, while walk-length does. 
See
\figref{hexagons}
for an example with $6$ vertices. 
To illustrate the difference more concretely in $\widetilde{G}_6$, consider the Dowker sink
filtration $\Dsi_\delta$. 
The simplex $[1,5,6]$ (with vertex $6$ as sink) is
added to the Dowker complex at $\delta=1$, while this simplex is not added
to walk-length filtration until $\delta=5$, which is the cheapest way to include these three vertices in a walk. 
For this reason, the 1-dimensional persistence class has the same death time in Dowker filtration but in walk-length it dies at a later $\delta$.
Note that in this example, however, both Dowker and walk-length persistence differentiate $G_n$ and $\widetilde{G}_n$ for 2-dimensional homology. 
\begin{figure}
    \centering
    \includegraphics[width=0.9\textwidth]{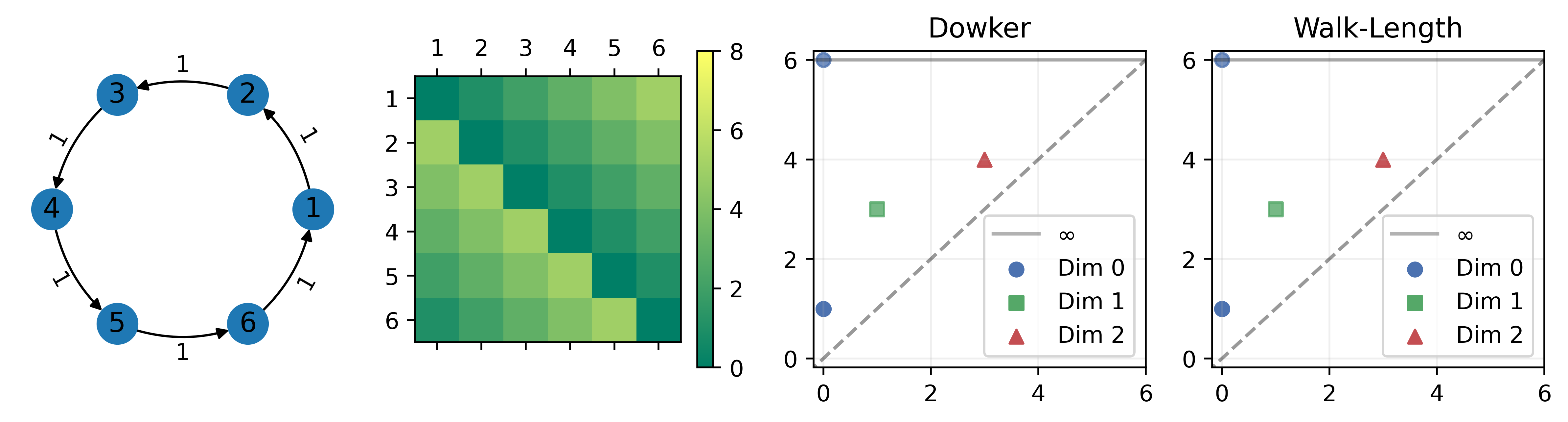}
    \includegraphics[width=0.9\textwidth]{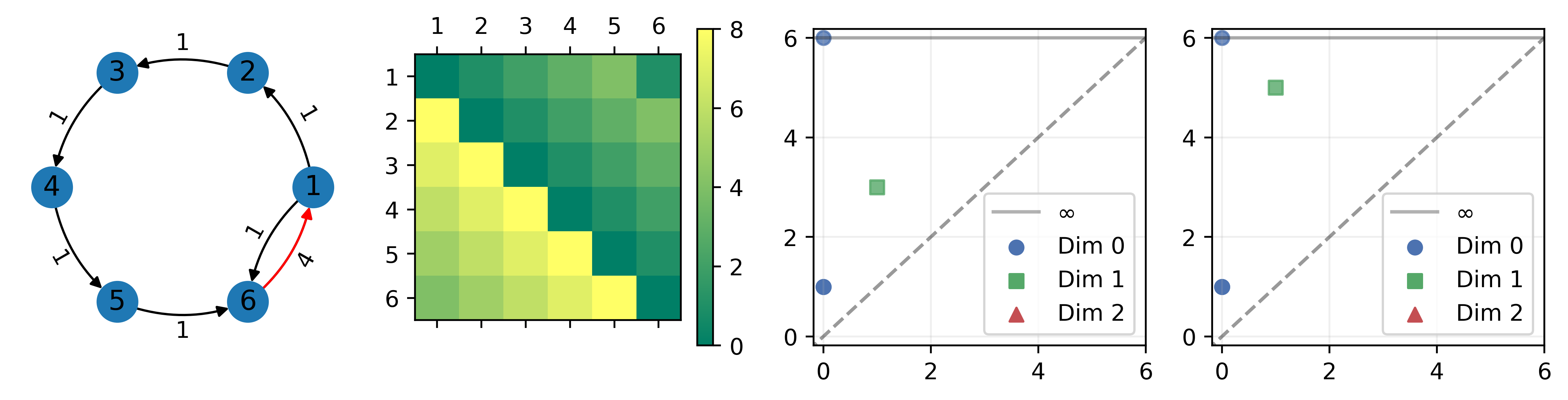}
    \caption{Top left shows the cycle graph $D_6$; bottom left shows the modified graph $\widetilde{D}_6$ where the weight of $(6,1)$ is equal to $4$. On the second column we show the shortest distance matrices, illustrating the weights in the networks $G_6$ and $\widetilde{G}_6$, respectively. Third and fourth columns show the corresponding persistence diagrams.
    }
    \label{fig:hexagons}
\end{figure}
\end{example}

In general, from a topological viewpoint, the early addition of the simplex $[x_{n-1},x_n,x_1]$ to the Dowker filtration makes the homology of $\widetilde{G}_n$ equivalent to that of $G_{n-1}$.
In contrast, the walk-length filtration is completely changed by the higher weight $\omega(x_n,x_1)$ since it breaks the cyclic sequence of edges with weight $1$.
In other words, we see that walk-length filtration is more sensitive to non-directed cycles, since complete directed cycles will die sooner.

\subsection{Hippocampal networks}
\label{ssec:HippocampalNetworks}

In order to see the accuracy and efficiency of walk-length persistence, we replicate the experiment on simulated hippocampal networks developed in \cite[Section 7]{CM2018-Dowker}, where Dowker persistence is employed.
The data used in this experiment simulates hippocampal activity of an animal walking randomly in arenas with different number of holes; a network is induced from each trial and the goal is to classify these networks by number of holes.
This experiment arises from findings in neuroscience \cite{OKeefe1971} that show a correspondence between the animal exploring specific spatial regions, or \textit{place fields}, and increased activity in specific hippocampal cells, or \textit{place cells}.
The dataset generated is described as follows.

Five arenas are set to be squares with side length $L=10$; each has 0, 1, 2, 3, and 4 holes (or obstacles) respectively.
For each arena, there were $20$ simulations, where a random walk of $5000$ steps was generated in a square grid with step length $0.05L$; in each step, the next direction is uniformly chosen at random between all allowed directions, that is, excluding borders and holes.
Then, for each trial $\ell_k$, a collection of $n_k$ locations, the place fields, is allocated randomly inside the corresponding arena as balls of radius $0.05L$, where $n_k$ is a random integer between $150$ and $200$.
Denote the set of time steps as $T=\{1,2,\dots,5000\}$. The trajectory information is registered with spike functions $r_i:T\rightarrow\{0,1\}$ for each place cell $i$, $1\leq i\leq n_k$, such that 
\[
r_i(t) = \begin{cases}
    1,& \text{if trajectory intersects place field $i$ at time $t$;}\\
    0,& \text{otherwise.}
\end{cases}
\]
Figure~\ref{fig:sample-arenas} shows two sample walk simulations $\ell_k$, together with the $n_k$ place fields and the resulting spike data $\{r_i(t)\}_{i,t}$ of place cells.
\begin{figure}
    \centering
    \includegraphics[width=0.85\textwidth]{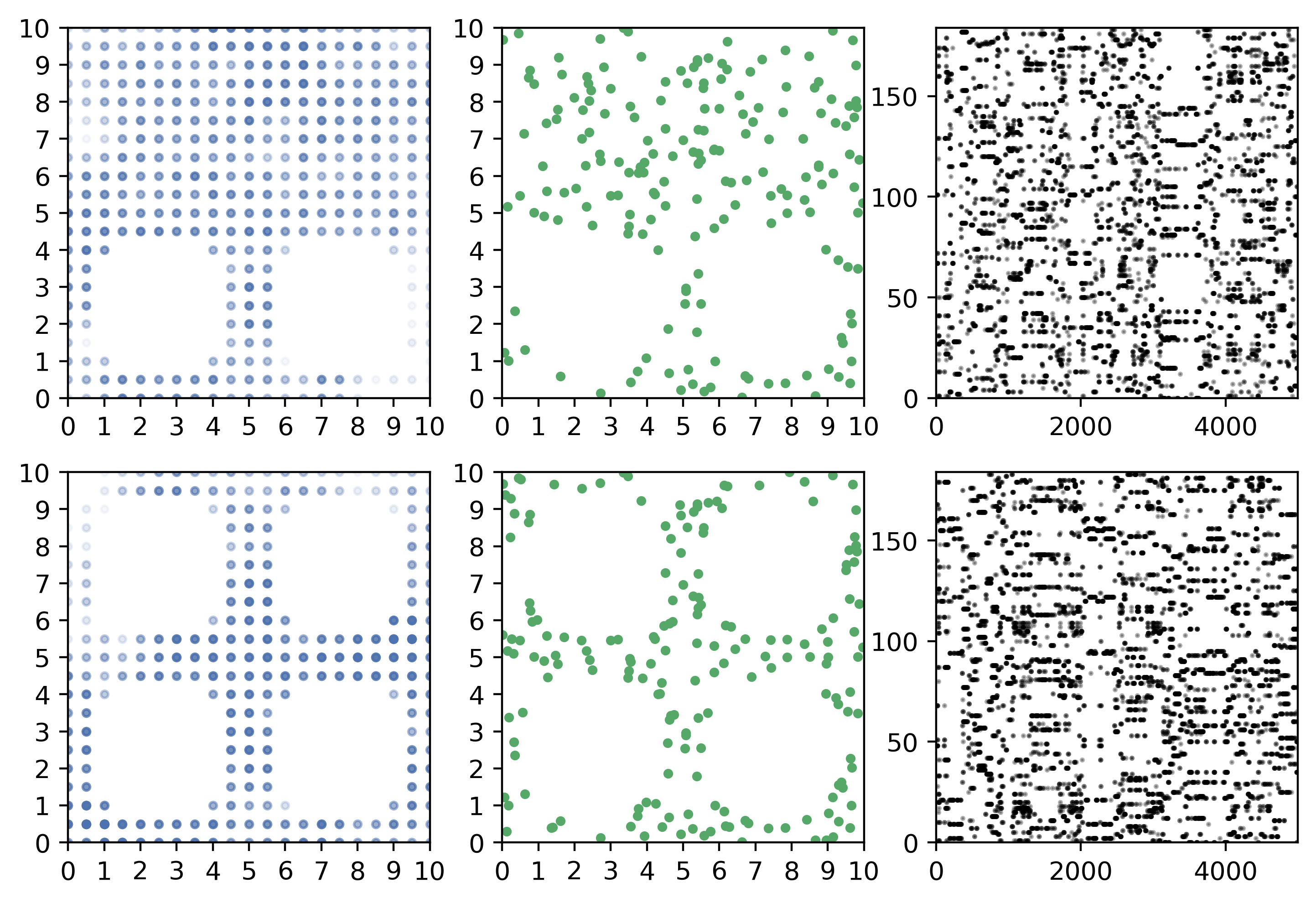}
    \caption{Sample arenas with two (top row) and four (bottom row) holes. Left column shows walks along the grid; shades generated with most visited points. Middle column shows the centers of place fields. Right column shows the corresponding spike pattern on place cells; the $x$-axis represents time steps and the $y$-axis represents the $n_k$ place cells.}
    \label{fig:sample-arenas}
\end{figure}

Now, given the simulated data, networks $(X_k,\omega_k)$ are defined to encode each trajectory. For trial $\ell_k$, the set of vertices $X_k$ in the network is given by the set of place cells, and the weight between two vertices $x_i$ and $x_j$ is defined as
\[
\omega_k(x_i,x_j) := 1-\frac{N_{i,j}(5)}{\sum_{i=1}^{n_k} N_{i,j}(5)}
\]
where
\[
N_{i,j}(5) := \#\{(s,t)\in T^2 \;:\; t\geq2,\;\;
1\leq t-s\leq 5, \;\; r_i(s)=r_j(t)=1\}.
\]
Essentially, this weight function 
measures the complement of how frequently the activation of cell $x_j$ was preceded by $x_i$ within a time span of $5$ steps.
In this way, the more probable it is to go from $x_i$ to $x_j$, the closer the weight will be to zero. For instance, if two place cells match to two place fields that are more than 5 steps apart in the arena, then the weight of the edge will be equal to 1.
The reason for using the complement is that when computing a Dokwer or walk-length filtration, we want the most common connections to appear as simplices earlier, while less related cells will be connected closer to the end of the filtration.

After preprocessing, the results obtained in \cite{CM2018-Dowker} 
show a good classification obtained from Dowker persistence, where the 4, 3, and 2-hole arenas are well separated into clusters by bottleneck distance between persistent diagrams at a certain threshold. Additionally, they show the contrast with an unsuccessful classification using Rips persistence (obtained from the symmetrization of the networks), where the results show how the directions in these networks are necessary to identify number of holes. The dendrograms from that study are reproduced in Appendix~\ref{sec:dendrograms}.

We conducted a thorough analysis using Rips, Dowker, and walk-length persistence and with modifications in preprocessing; these modifications are described as follows.
For the results in \cite{CM2018-Dowker}, all edge weights were shifted by the minimum weight in each network, omitting the edges with weight $1$. That is, each $\omega_k$ is replaced by $\omega_k^{\min,1}$, defined as
\[\omega_k^{\min,1}(x_i,x_j) := \begin{cases}
\omega_k(x_i,x_j) -m_k, & \omega_k(x_i,x_j)<1; \\
1, & \omega_k(x_i,x_j)=1;
\end{cases}
\qquad \text{where } m_k = \min_{i,j}\omega_k(x_i,x_j).
\]
Shifting by the minimum value $m_k$ is not uniform across the data set, and since we want to compare the networks, this shift might lose some useful information to compare the length of the smaller loops in the networks. Hence, for further analysis, we also consider the weight functions given by
\[
\omega_k^{\min,\text{purge}}(x_i,x_j) := \begin{cases}
\omega_k(x_i,x_j) -m_k, & \omega_k(x_i,x_j)<1; \\
\infty, & \omega_k(x_i,x_j)=1.
\end{cases}
\]
and
\[
\omega_k^{\text{purge}}(x_i,x_j) := \begin{cases}
\omega_k(x_i,x_j), & \omega_k(x_i,x_j)<1; \\
\infty, & \omega_k(x_i,x_j)=1.
\end{cases}
\]
The term \textit{purge} is used to indicate the deletion of all edges with original weight equal to 1. 
Such edges correspond to pairs of place cells that never fire subsequently, hence it is reasonable to omit these edges from the networks by setting their weight to $\infty$. 
Note that purged networks will no longer be complete graphs but these are still a valid input for walk-length persistence. 
In the case of Dowker, we can also compute the filtration either directly or by using their associated shortest-distance weights.
The reason for purging is that since all edges have weight less than or equal to 1, all sets of three vertices using the original weights will form a 2-simplex at filtration value $2$ in walk length persistence, potentially preventing the tracking of features that involve longer walks.
We also included Rips persistence with symmetrized versions of these networks, both by min-symmetrization, that is, replacing the weights 
$\omega_k(x_i,x_j)$ with $\min(\omega_k(x_i,x_j),\omega_k(x_j,x_i))$, 
and by max-symmetrization similarly. 
We will not use additional notation for these weights for the sake of simplicity.

After applying persistence, in each case the pairwise bottleneck distances between persistence diagrams are computed. 
We show the resulting MDS (Multidimensional Scaling) plots in Figure~\ref{fig:mds-plots}, where the number of holes in the arena is given by symbol/color.
In this figure, we show the MDS plots for Dowker (top row) and walk length (bottom row), using either the weights $\omega_k^{\min{},1}$ (left column) or $\omega_k^{\text{purge}}$ (right column). 
The symbols for each data point represent the number of 
In this case, qualitatively we see that that the walk length persistence with the purge weight function (bottom right) gives the clearest stratification of the number of holes. 

Thus, to give a quantitative view of the accuracy, we set up a 4-nearest neighbor classifier, so that a new persistence diagram gets a predicted label of the most common of the label from the four closest neighbors in the data set. 
We report the accuracy obtained using this classifier  with leave-one-out cross-validation and  show our findings in Table~\ref{tab:accuracy-results}. 
Amongst all the tests, using $\omega_k^{\text{purge}}$ gives the best accuracy of 86\%.
The results of this test which are equivalent to the tests run in \cite{CM2018-Dowker} is standard Dowker with $\omega_k$, which results in 50\% accuracy. 
However, we do note that different choices of encoding of the edge weight information result in improved versions of the classifier. 
In particular, the results using the setup from \cite{CM2018-Dowker} are given by the row Dowker and column $\omega_k^{\min, 1}$ as 71\%. 
In this case, we see that the classification using the filtration can actually be improved 82\% using a the weight function $\omega_k^{\text{purge}}$. 
However, the overall best classification comes from using the walk-length filtration with the $\omega_k^{\text{purge}}$ weight function. 
These results further emphasize that both Dowker and Walk-Length persistence are sensitive to the choice of function used.

\begin{figure}
    \centering
    \includegraphics[width=0.4\textwidth]{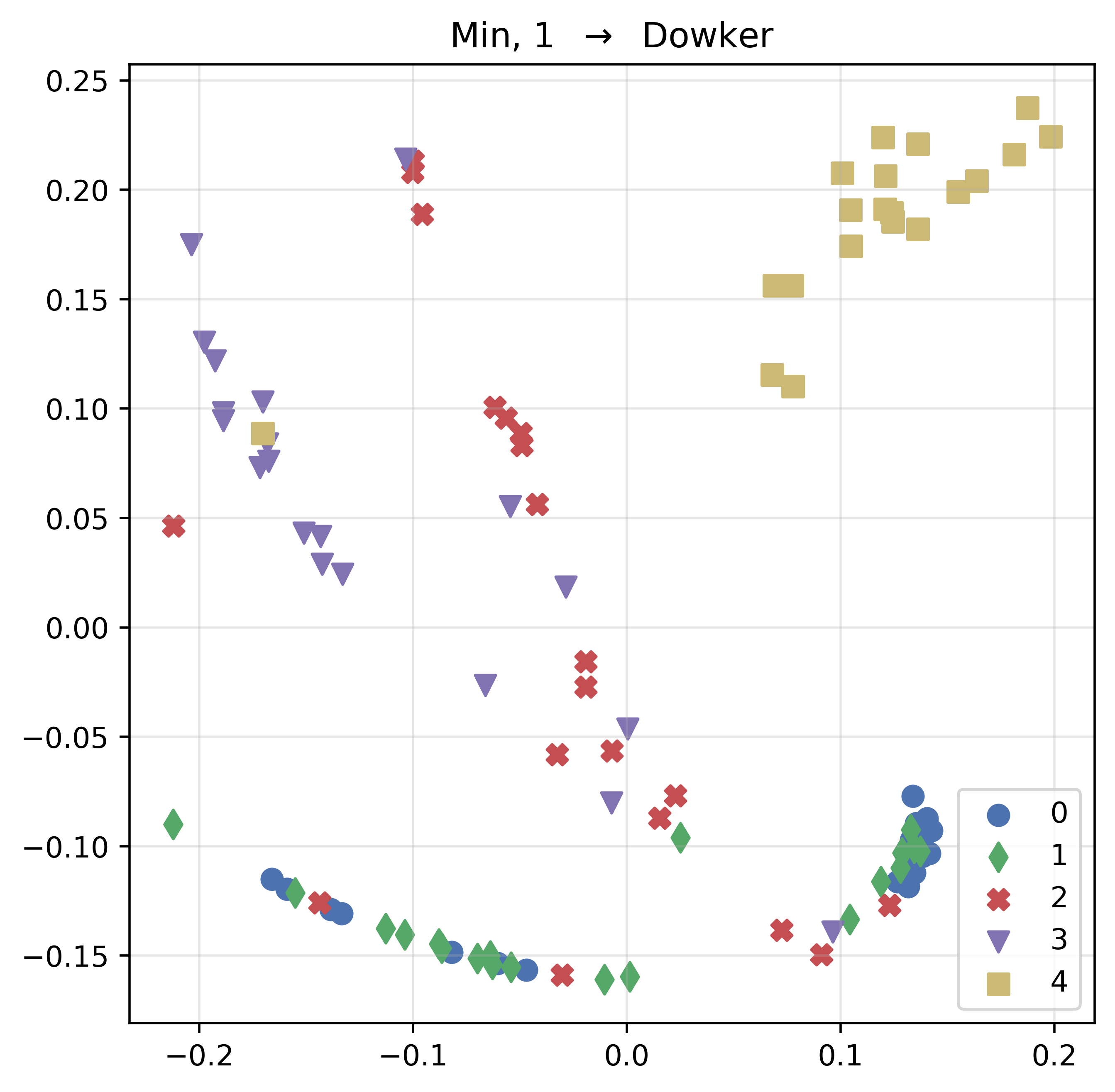} \qquad
    \includegraphics[width=0.4\textwidth]{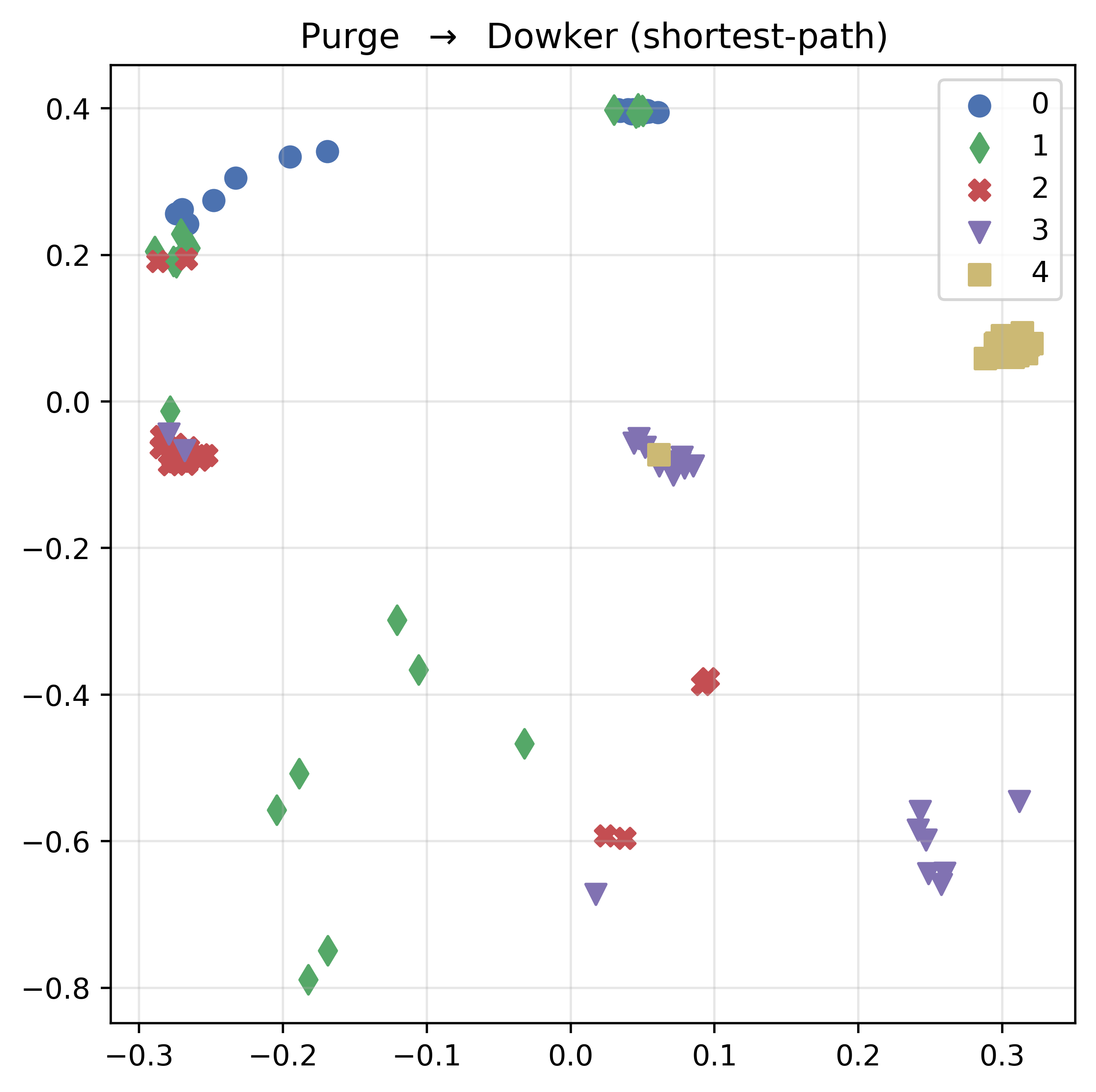}
    \includegraphics[width=0.4\textwidth]{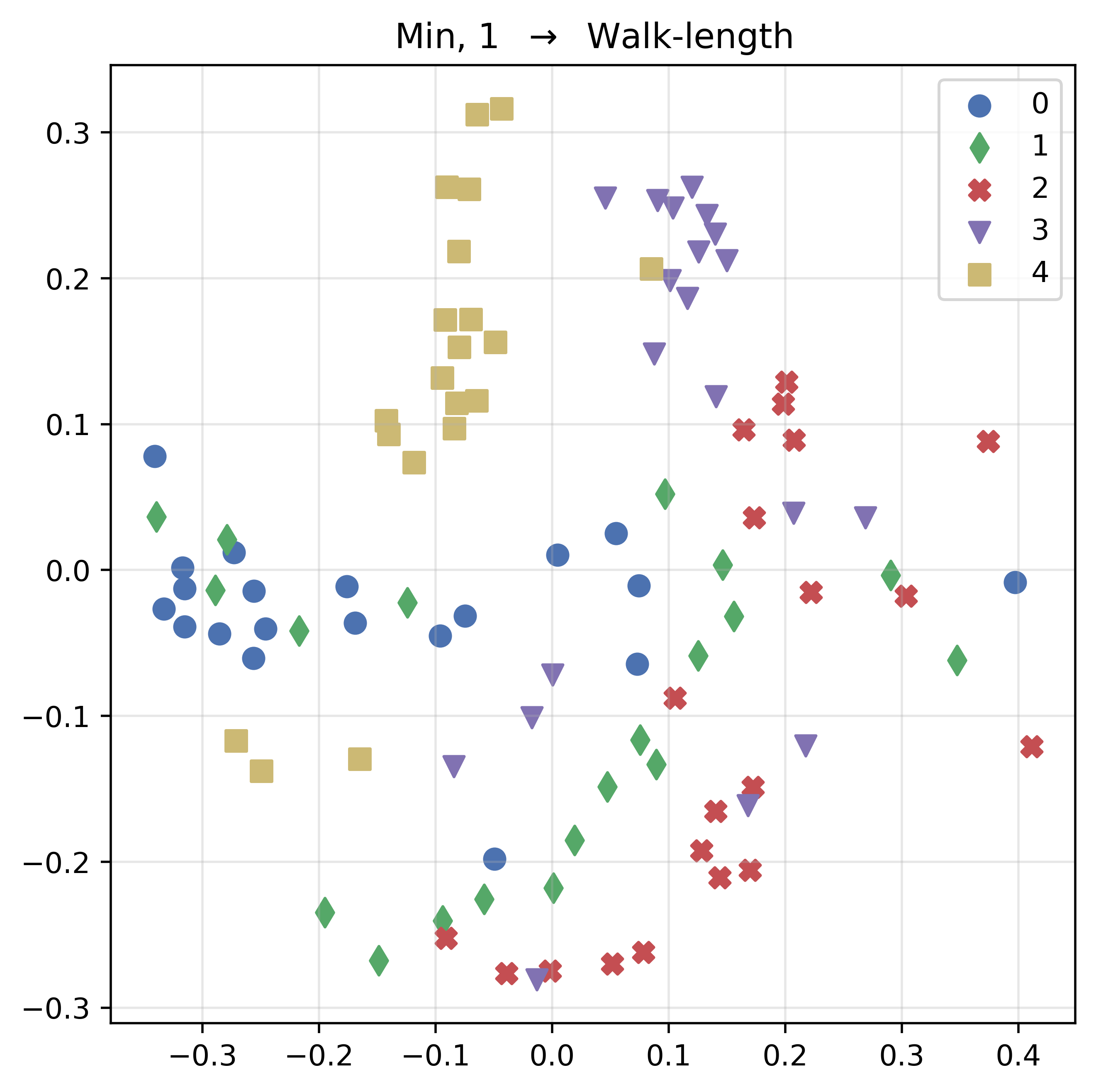} \qquad
    \includegraphics[width=0.4\textwidth]{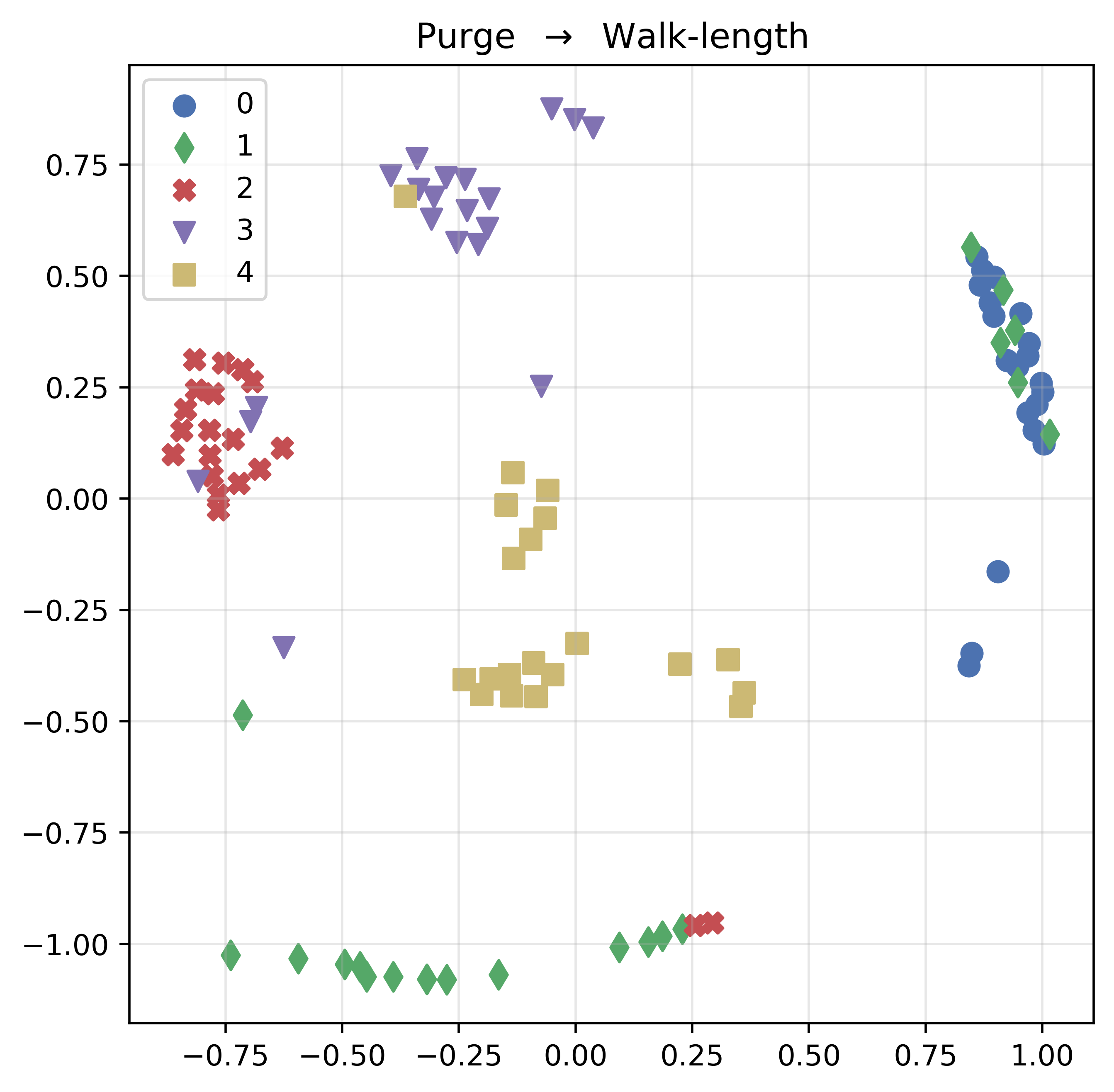}
    \caption{Multidimensional scaling plots from bottleneck distances using Dowker (top row) and walk-length (bottom row) persistence via the preprocessed weights $\omega_k^{\min{},1}$ (left column) and $\omega_k^{\text{purge}}$ (right column).}
    \label{fig:mds-plots}
\end{figure}

\begin{table}
\centering
\begin{tabular}{|l|c|c|c|c|c||c}
\hline & & & & \\[-2.2ex]
 & $\omega_k$ & $\omega_k^{\min,1}$ & $\omega_k^{\min,\text{purge}}$ & $\omega_k^{\text{purge}}$ 
 \\[-2.2ex] & & & & \\ \hline
Rips (min-sym)         & 37\% & 56\% & 63\% & 73\% \\ \hline
Rips (max-sym)         & 57\% & 51\% & 51\% & 59\% \\ \hline
Dowker              & 50\% & 71\% & 73\%$^\dagger$ & 82\%$^\dagger$ \\ \hline
Dowker (Shortest)   & 49\% & 71\% & 71\% & 82\% \\ \hline
Walk-Length         & 7\% & 80\% & 76\% & \textbf{86\%} \\ \hline
\end{tabular}
\caption{Accuracy percentages for different preprocessing methods. 
These are Rips filrations using either min or max symmetrization; standard Dowker filtration; Dowker of the shortest path network; and the walk-length filtration.
In the cases marked with $\dagger$, since input graphs are not complete, Dowker persistence diagrams might have different numbers of infinite classes, resulting in some pairs of diagrams having infinite distance.}
\label{tab:accuracy-results}
\end{table}

An important note is that the Dowker filtration with $\omega_k^{\text{purge}}$ does not converge to a complete simplicial complex, since the networks $(X_k,\omega_k^{\text{purge}})$ are not complete, and thus some 1-dimensional persistence points will have infinite persistence. 
In such cases, the bottleneck distance between two diagrams will be infinite if they have a different number of points at infinity.
In other words, the classification in that case is ultimately determined by the number of loops present in the complex at the end of the filtration.

We also show the confusion matrices obtained from the $k$NN classifier in Figure~\ref{fig:confusionmats}.
\begin{figure}
    \centering
    \includegraphics[width=0.4\textwidth]{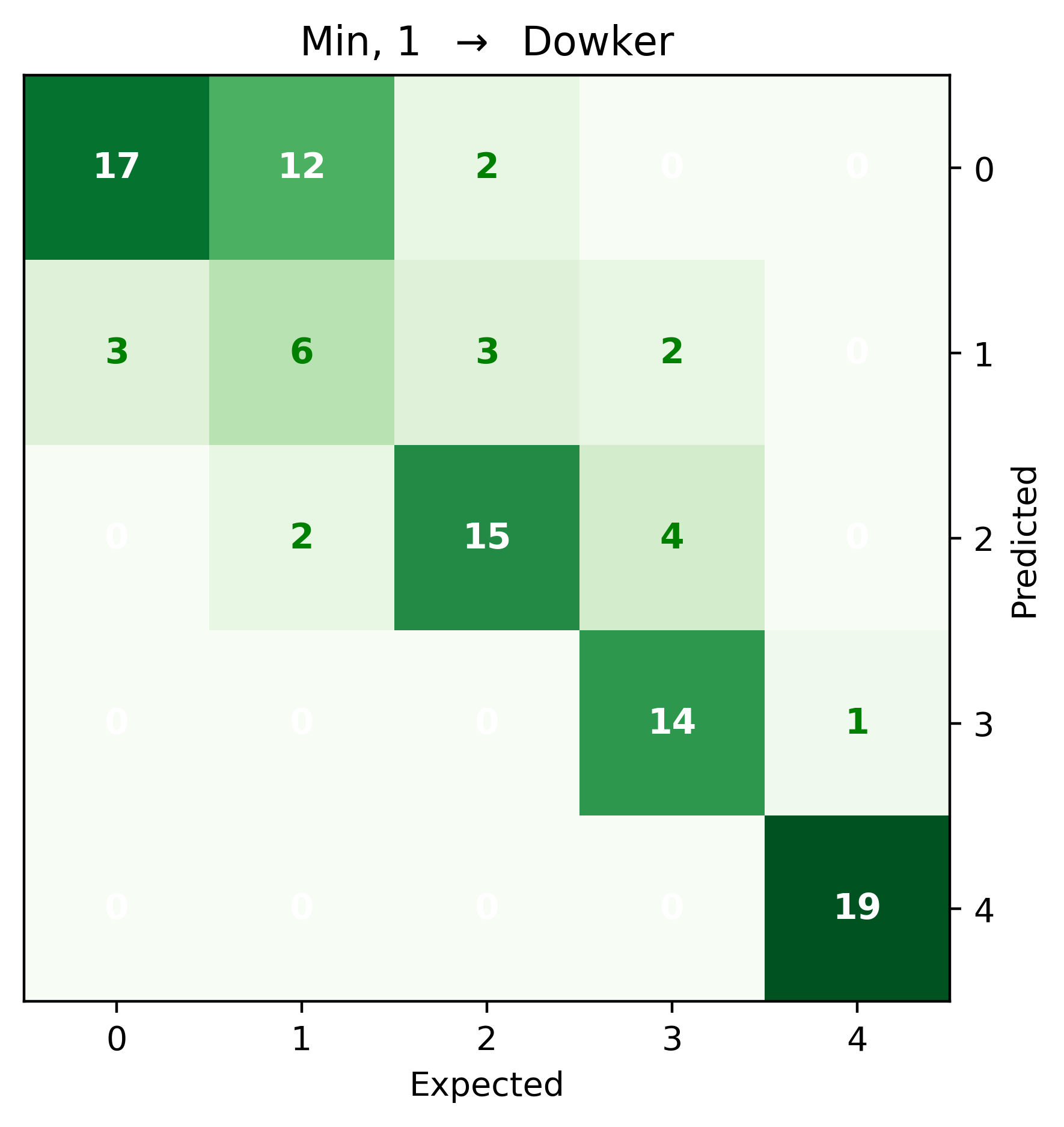} \qquad
    \includegraphics[width=0.4\textwidth]{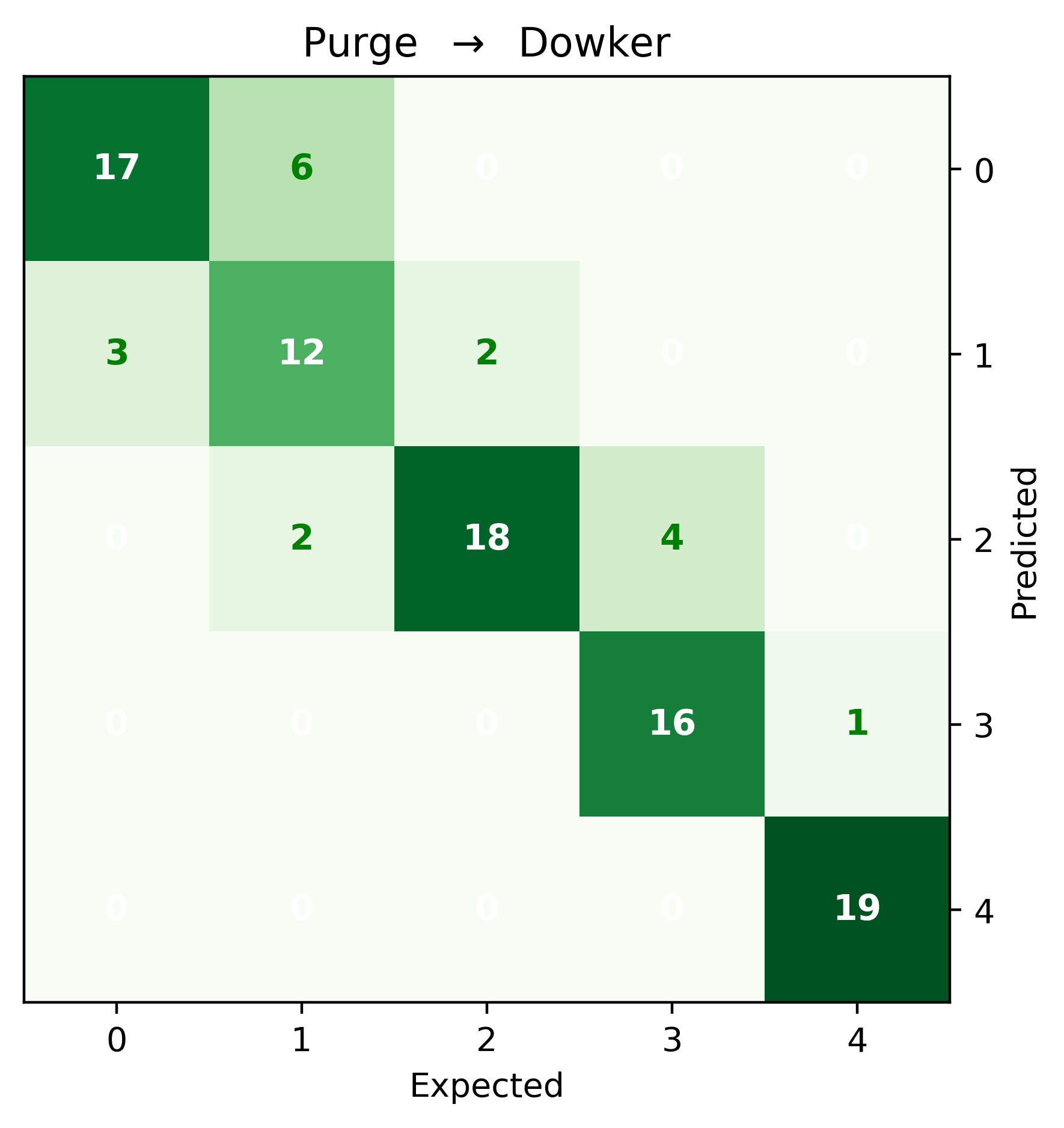}
    \includegraphics[width=0.4\textwidth]{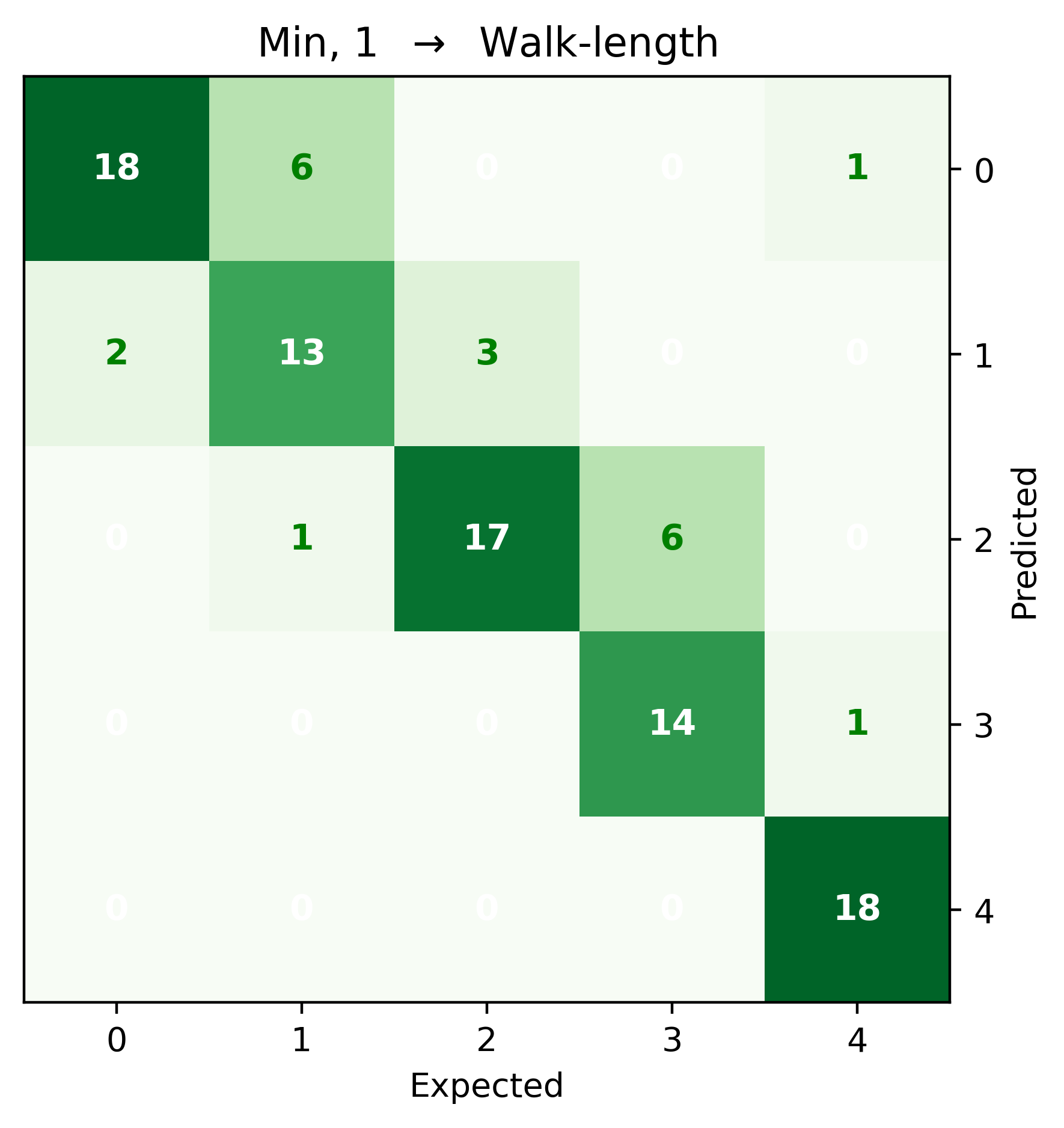} \qquad
    \includegraphics[width=0.4\textwidth]{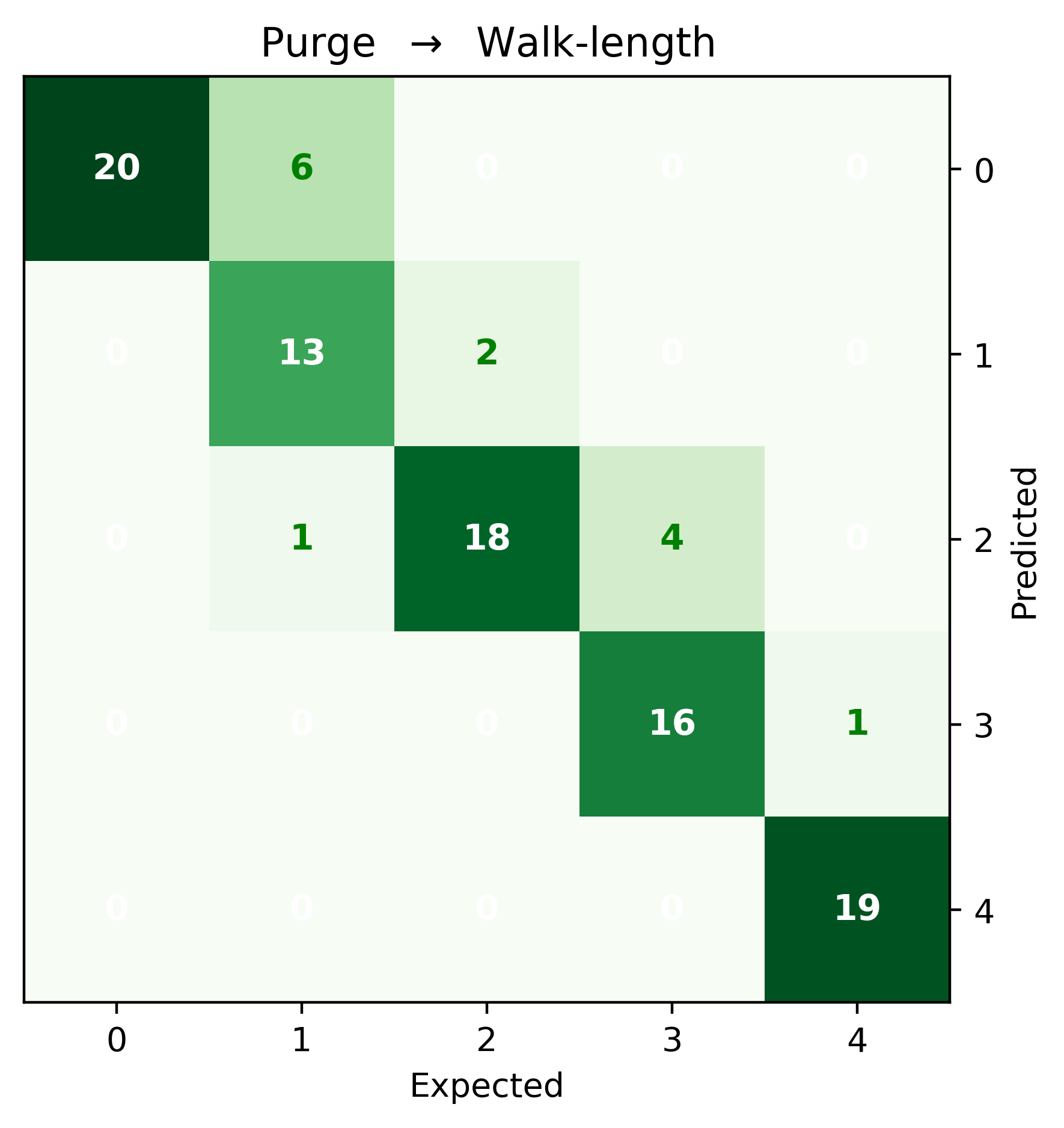}
    \caption{Confusion matrices from classification results using Dowker (left) and Walk-Length (right) persistence on the networks $(X_k,\omega_k^{\text{purge}})$.}
    \label{fig:confusionmats}
\end{figure}
The main difference is seen in the simulations performed in the arena without holes, where Dowker is identifying one persistent loop in $3$ of these networks, leading to misclassification. There is also a considerable improvement from the initial setting (Dowker persistence with weights $\omega_k^{\text{min},1}$) in the simulations with one hole, gained either by replacing the weights with $\omega_k^\text{purge}$ or by using walk-length persistence.

\textbf{Note on computation time for simplicial filtrations.} 
Most computations were run on an Intel(R) Xeon(R) CPU E5-2680 v4 @ 2.40GHz model.
Average recorded run time for building the $2$-dimensional simplicial Dowker $\delta$-sink filtration of each network was $8.70\pm0.31$ seconds, while for walk-length filtration it was $2.79\pm0.18$ seconds. 
Recall that these networks are complete weighted directed graphs with number of nodes between 150 and 200. Also note that these times do not include persistence, which was computed with the Gudhi library \cite{gudhi:FilteredComplexes}.
\section{Discussion and Future Work}

In this paper, we have defined the \textit{walk-length filtration}. 
We investigated stability and computability of the resulting persistence diagrams.
Finally, we provided experiments where it was used to differentiate input directed graphs in several example settings. 
We next note interesting further directions for this work.

In the course of showing stability, we gave a new definition for the distance between networks (that is, complete weighted digraphs) and showed that the resulting persistence diagrams are stable under the bottleneck distance with respect to this new distance. 
However, there is something unsatisyfing in this result which leads to many interesting open questions. 
First, defining the network distance without restricting to isomorphisms of the vertex set resulted in a multiplicative constant in the stability theorem directly related to the size of the inputs. 
In some sense, this should not be surprising since we are relating a $L_\infty$ type distance (bottleneck) with an $L_1$ type distance, and so summing over some maximal value will result in such a constant.
However, our proof method is directly tied to being able to work with interleavings of the resulting persistence modules, and thus is restricted to the bottleneck distance.
It would be interesting to see if a version of stability could be given instead by comparing the persistence diagrams with Wasserstein distance. 
For example, \cite{Skraba2020} gives stability of this form by comparing two input filtrations $f,g: K \to \R$ using the distance $\sum_{\sigma \in K} |f-g|$. 
However, since we start with only graph information as input, it is not immediately clear how to relate the network distance with this distance for the entire simplicial complex. 

Another interesting direction of work is to gain further understanding on the relationship between the walk-length persistence and other available methods for analyzing digraph inputs from the TDA literature. 
In particular, the definition of persistent path homology (PPH), arises from constructing a homology theory tightly bound to small paths in the input directed graph. 
However, a major difference is that given an input network, the paths to be studied are those for which all edges have weight bounded by $\delta$ (again, in an $L_\infty$-style construction) rather than requiring the entire length of the path ($L_1$-style) to be bounded by a particular weight. 
It would be interesting to see if these two definitions are more closely related than it would seem.

\subsection*{Acknowledgements}

This work was supported in part by the National Science Foundation through grants DGE-2152014, CCF-2106578, and CCF-2142713.
The authors would like to thank Brittany Fasy and Facundo M\'emoli for helpful discussions during the course of this work. 

\appendix

\bibliographystyle{plainurl}%
\bibliography{references}

\section{Bounding the running time}
This lemma is used in Section~\ref{ssec:algorithm} to bound the running time. The method largely follows a post\footnote{\url{https://mathoverflow.net/questions/17202/sum-of-the-first-k-binomial-coefficients-for-fixed-n}}  on stack overflow. 
\begin{lemma}
\label{lem:binom_bound}
For a fixed $K$, 
\begin{equation*}
    \sum_{i = 0}^{K} \binom{N}{i} \leq \binom{N}{K}\frac{N-K+1}{N-2K+1}. 
\end{equation*}
\end{lemma}

\begin{proof}
    For a fixed $K$ and $N \to \infty$, 
    \begin{align*}
        \frac{1}{\binom{N}{K} } &\left(\binom{N}{K} + \binom{N}{K-1} + \binom{N}{K-2} \cdots  \right)\\
        &=
        1 + \frac{K}{N-K+1} + \frac{K(K-1)}{(N-K+1)(N-K+2)} + \frac{K(K-1)(K-2)}{(N-K+1)(N-K+2)} + \cdots \\
        & \leq 1 + \frac{K}{N-K+1} + \left(\frac{K}{N-K+1} \right)^2 + \left(\frac{K}{N-K+1} \right)^3 + \cdots \\
        & = \frac{1}{1 - \frac{K}{N-K+1}}
        = \frac{N-K+1}{N-2K+1}
    \end{align*}
    where the last line uses the equation for a geometric series after checking that $|\frac{K}{N-K+1}| <1$. 
\end{proof}

\section{Dendrograms}\label{sec:dendrograms}

In Figure~\ref{fig:dendrograms} we show reproduced dendrograms of \cite{CM2018-Dowker}. These show the results of single linkage clustering with respect to bottleneck distance between persistence diagrams of all networks. 
Entries are labeled with the number of O's corresponding to the number of holes in the arena generating the data point. 
No label indicates the arena with no holes.
We show the same analysis corresponding to walk-length persistence in Figure~\ref{fig:dendrogram-WL}.
In these cases, we see a qualitative improvement in clustering for the walk length persistence over the Dowker persistence, which mirrors the quantitative comparison done in Section~\ref{ssec:HippocampalNetworks}.
\begin{figure}
    \centering
    \includegraphics[width=0.49\textwidth]{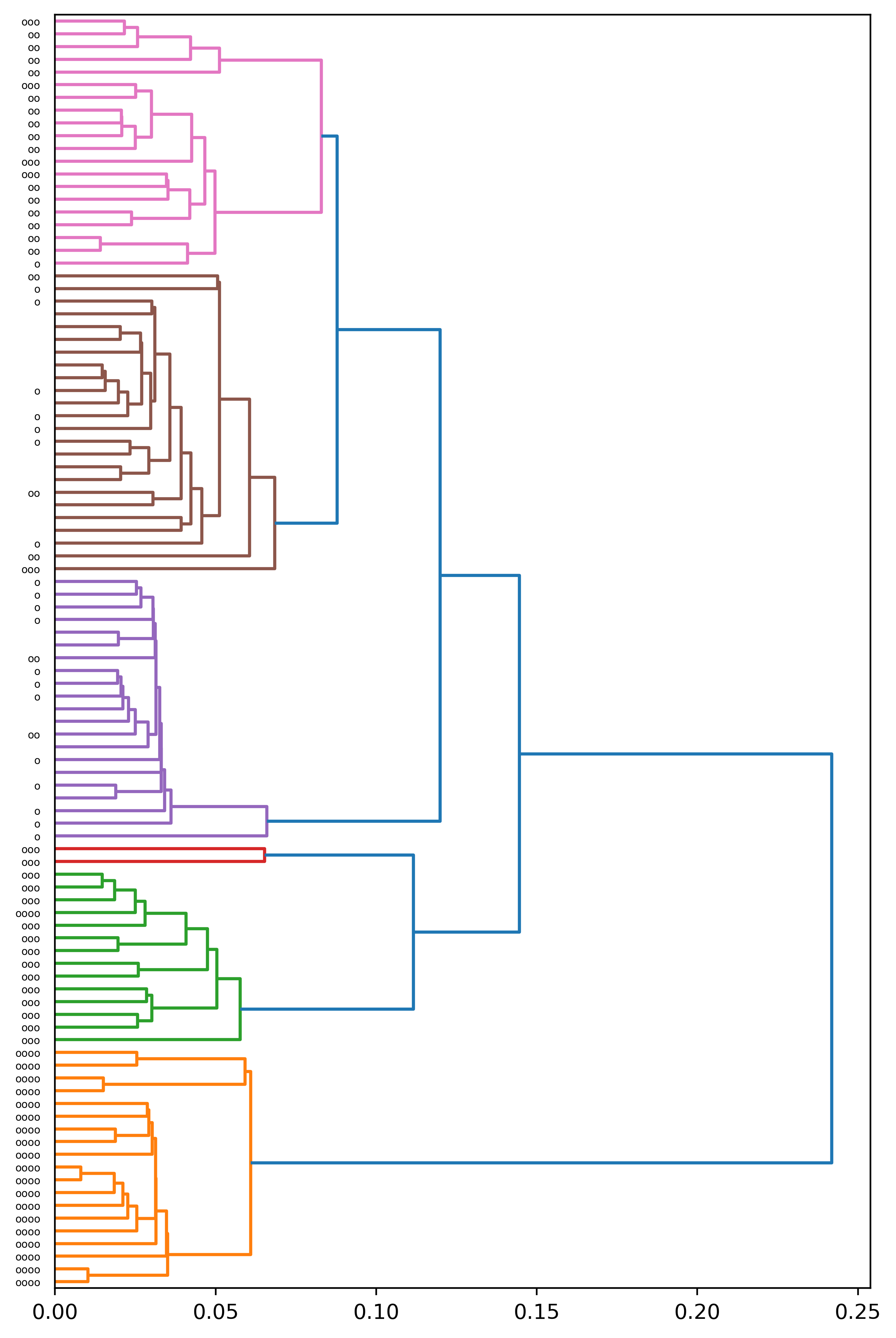}
    \includegraphics[width=0.49\textwidth]{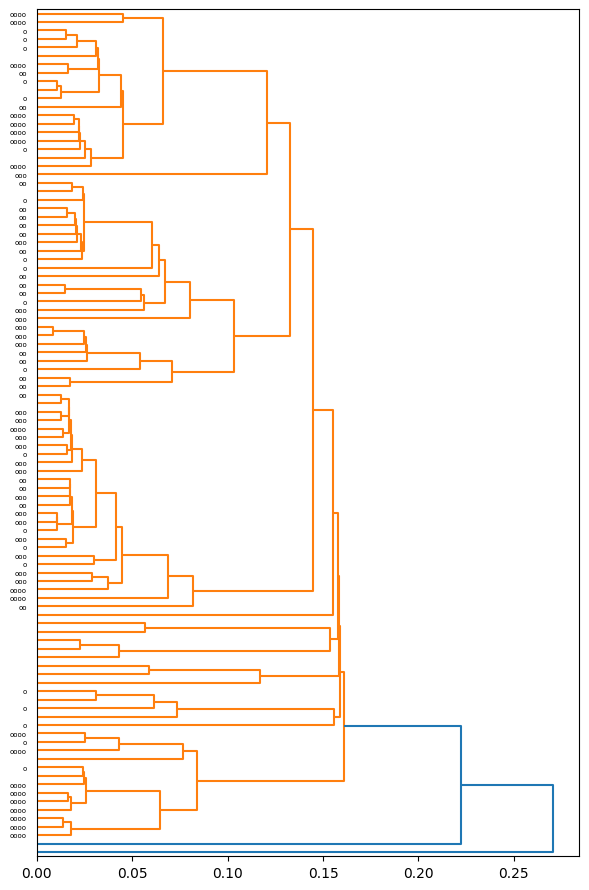}
    \caption{Dendrograms from bottleneck distance between persistence diagrams obtained from directed networks. Left: Dowker (asymmetric) persistence; colors represent clustering at threshold 0.085. Right: Rips persistence after max-symmetrization.}
    \label{fig:dendrograms}
\end{figure}

\begin{figure}
    \centering
    \includegraphics[width=0.8\textwidth]{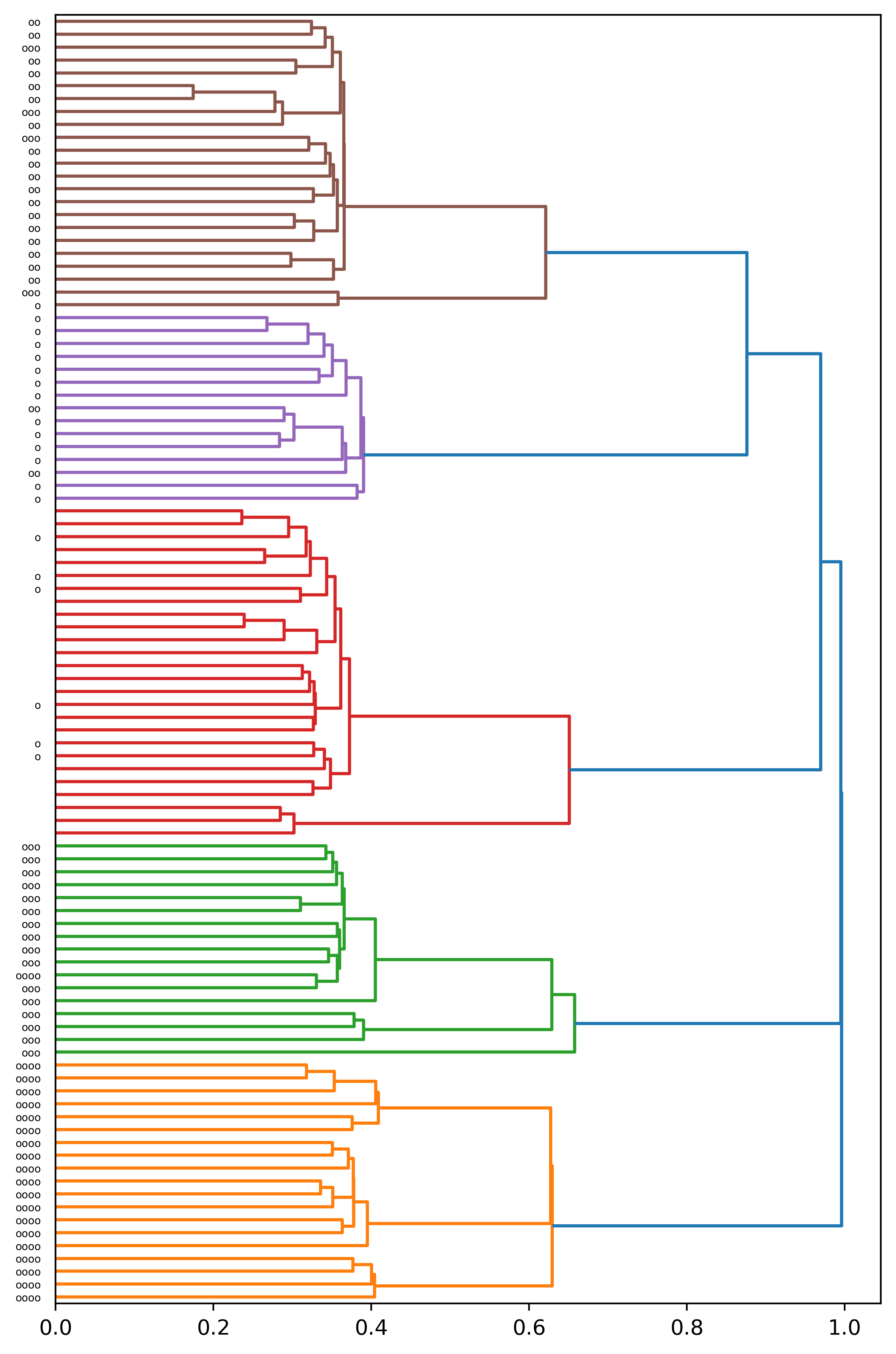}
    \caption{Dendrogram from bottleneck distance between walk-length persistence diagrams obtained from directed networks with weight $\omega_k^{\text{purge}}$. Labels on vertical axis indicate number of holes on each trial, where arenas with no holes do not have a label. Colors represent clustering at threshold 0.8.}
    \label{fig:dendrogram-WL}
\end{figure}

\end{document}

%% file: sections/abstract.tex
Directed graphs arise in many applications where computing persistent
homology helps to encode the shape and structure of the input information.
However, there are only a few ways to turn the directed graph information into an undirected simplicial complex filtration required by the standard persistent homology framework.
In this paper, we present a new filtration constructed from a directed graph, called the walk-length filtration. 
This filtration mirrors the behavior of small walks visiting certain collections of vertices in the directed graph. 
We show that, while the persistence is not stable under the usual
$L_\infty$-style network distance, a generalized $L_1$-style distance is,
indeed, stable. 
We further provide an algorithm for its computation, and investigate the
behavior of this filtration in examples, including cycle networks and synthetic hippocampal networks with a focus on comparison to the often used Dowker filtration.

%% file: figures/tikz/shortest_dist_graph.tex
\begin{tikzpicture}[>={Latex[length=2mm]},scale=1]
    \node[minimum size=0.5cm,draw,circle,inner sep=1pt] (a) at (0,2) {$a$};
    \node[minimum size=0.5cm,draw,circle,inner sep=1pt] (b) at (3,2) {$b$};
    \node[minimum size=0.5cm,draw,circle,inner sep=1pt] (c) at (1.5,-0.2) {$c$};
    \draw[->] (a) edge["$1$", bend left=9] (b);
    \draw[->] (b) edge["$1$", bend left=9] (c);
    \draw[->] (c) edge["$1$", bend left=9] (a);
    \draw[->] (b) edge["$10$", bend left=9] (a);

    \begin{scope}[shift={({6},{0})}]
    \node[minimum size=0.5cm,draw,circle,inner sep=1pt] (a) at (0,2) {$a$};
    \node[minimum size=0.5cm,draw,circle,inner sep=1pt] (b) at (3,2) {$b$};
    \node[minimum size=0.5cm,draw,circle,inner sep=1pt] (c) at (1.5,-0.2) {$c$};
    \draw[->] (a) edge["$1$", bend left=9] (b);
    \draw[->] (b) edge["$2$", bend left=9] (a);
    \draw[->] (b) edge["$1$", bend left=9] (c);
    \draw[->] (c) edge["$2$", bend left=9] (b);
    \draw[->] (c) edge["$1$", bend left=9] (a);
    \draw[->] (a) edge["$2$", bend left=9] (c);
    \end{scope}
\end{tikzpicture}

%% file: figures/tikz/unstability_example.tex
\begin{tikzpicture}[>={Latex[length=2mm]},scale=1]
    \begin{scope}[shift={({0},{4})},scale=1.7]
        \node at (0,1.2) {\Large $D$};

        \node[circle,draw,inner sep=1.3pt] (a) at (0,0) {$a$};
        \node[circle,draw,inner sep=1pt] (b) at (1,1) {$b$};
        \node[circle,draw,inner sep=1.3pt] (c) at (2,0) {$c$};
        \draw[->,thick] (a) edge["1"] (b);
        \draw[->,thick] (b) edge["1"] (c);
        \draw[->,thick] (a) edge["1"] (c);

        \draw[->] (c) edge["10", bend left=30] (a);
    \end{scope}
    
    \begin{scope}[shift={({0},{0})},scale=1.7]
        \node at (0,1.2) {\Large $D_\epsilon$};

        \node[circle,draw,inner sep=1.3pt] (a) at (0,0) {$a$};
        \node[circle,draw,inner sep=1pt] (b) at (1,1) {$b$};
        \node[circle,draw,inner sep=1.3pt] (c) at (2,0) {$c$};
        \draw[->,thick] (a) edge["$1+\epsilon$"] (b);
        \draw[->,thick] (b) edge["$1+\epsilon$"] (c);
        \draw[->,thick] (a) edge["$1+\epsilon$"] (c);

        \draw[->] (c) edge["10", bend left=30] (a);
    \end{scope}
    
    \node[anchor=west,inner sep=0] at (4.6,1.8) {\input{figures/tikz/filtration_nonstability}};

    \node[anchor=south west,inner sep=0] at (11.5,2.6) {\includegraphics[width=0.25\textwidth]{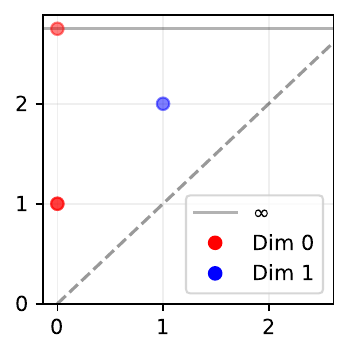}};

    \node[anchor=south west,inner sep=0] at (11.5,-1.5) {\includegraphics[width=0.25\textwidth]{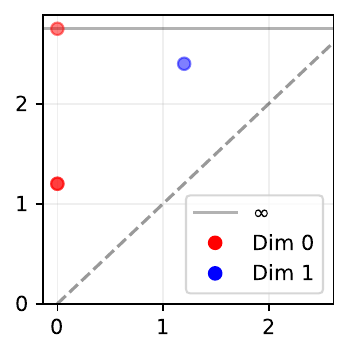}};
\end{tikzpicture}

%% file: figures/tikz/triangleineq_counterexample.tex
\begin{tikzpicture}[>={Latex[length=2mm]}]
    \node at (0,4) {\large{$\cX$}};
    \node[minimum size=0.5cm,draw,circle,inner sep=1pt] (a) at (0,3) {$x_1$};
    \node[minimum size=0.5cm,draw,circle,inner sep=1pt] (b) at (0,0) {$x_2$};
    \draw[->] (a) edge["$10$", bend left=10] (b);
    \draw[->] (b) edge["$1$", bend left=10] (a);

\begin{scope}[shift={({3.5},{0})}]
    \node at (0,4) {\large{$\cY$}};
    \node[minimum size=0.5cm,draw,circle,inner sep=1pt] (a) at (0,3) {$y_1$};
    \node[minimum size=0.5cm,draw,circle,inner sep=1pt] (b) at (0,0) {$y_2$};
    \draw[->] (a) edge["$5$", bend left=10] (b);
    \draw[->] (b) edge["$1$", bend left=10] (a);
\end{scope}

\begin{scope}[shift={({6},{0})}]
    \node at (2,4) {\large{$\cZ$}};
    \node[minimum size=0.5cm,draw,circle,inner sep=1pt] (a) at (0,3) {$z_1$};
    \node[minimum size=0.5cm,draw,circle,inner sep=1pt] (b) at (4,3) {$z_1'$};
    \node[minimum size=0.5cm,draw,circle,inner sep=1pt] (c) at (2,0) {$z_2$};
    \draw[->] (a) edge["0.1", bend left=10] (b);
    \draw[->] (b) edge["0.1", bend left=10] (a);
    \draw[->] (b) edge["5", bend left=10] (c);
    \draw[->] (c) edge["1", bend left=10] (b);
    \draw[->] (c) edge["1", bend left=10] (a);
    \draw[->] (a) edge["5", bend left=10] (c);
\end{scope}
\end{tikzpicture}